%% file: main.tex
\newif\ifTR
\TRtrue

\documentclass[cleveref,runningheads]{llncs}

\usepackage[T1]{fontenc}
\usepackage{graphicx}
\usepackage{enumitem}
\usepackage{amssymb}
\usepackage{amsmath}

\usepackage{amsthm}
\usepackage{hyperref}
\usepackage[capitalize]{cleveref}
\usepackage{tikz}
\usepackage{pgfplots}
\usepackage{xcolor}
\usepackage{caption}
\usepackage{subcaption}
\usepackage{xspace}
\usepackage{booktabs}
\usepackage{crimson}

\usetikzlibrary{myautomata,arrows,fit,shapes,backgrounds,extshapes,positioning,calc}

\tikzset{fit margins/.style={/tikz/afit/.cd,#1,
    /tikz/.cd,
    inner xsep=\pgfkeysvalueof{/tikz/afit/left}+\pgfkeysvalueof{/tikz/afit/right},
    inner ysep=\pgfkeysvalueof{/tikz/afit/top}+\pgfkeysvalueof{/tikz/afit/bottom},
    xshift=-\pgfkeysvalueof{/tikz/afit/left}+\pgfkeysvalueof{/tikz/afit/right},
    yshift=-\pgfkeysvalueof{/tikz/afit/bottom}+\pgfkeysvalueof{/tikz/afit/top}},
    afit/.cd,left/.initial=2pt,right/.initial=2pt,bottom/.initial=2pt,top/.initial=2pt}

\usepackage{lineno}
\linenumbers
\nolinenumbers

\def\orcidID#1{\smash{\href{http://orcid.org/#1}{\protect\raisebox{-1.25pt}{\protect\includegraphics{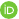}}}}}

\newcommand{\lang}{\mathcal{L}}
\newcommand{\langof}[1]{\lang \left(#1\right)}
\newcommand{\compl}[1]{\mathit{co}(#1)}
\newcommand{\elc}[1]{\widetilde{#1}} % = "element of complement" pouzivam pro oznaceni prvku komplementu (uvidi se, jaka znacka vypada nejlip)
\newcommand{\rev}[1]{\mathit{rev}(#1)} % reverse
\renewcommand{\det}[1]{\mathit{det}(#1)} % determinizace
\newcommand{\co}[1]{\mathit{co}(#1)} % komplement

\newcommand{\A}{\mathcal{A}}
\newcommand{\B}{\mathcal{B}}
\newcommand{\C}{\mathcal{C}}
\newcommand{\D}{\mathcal{D}}

\newcommand{\nat}[0]{\mathbb{N}}

\newcommand{\Atwo}[0]{\A_2}

\newcommand{\inportsof}[1]{\mathit{In}_{#1}} % mnozina vstupnich portu daneho automatu
\newcommand{\outportsof}[1]{\mathit{Out}_{#1}} % mnozina vystupnich portu daneho automatu

\newcommand{\gates}{\delta_{\mathit{trans}}} % prechody mezi dvema komponentami automatu

\newcommand{\tprod}{\delta^{\times}} % multi-product transition function
\newcommand{\tprodinter}{\gamma^{\times}} % intergalactic multi-product transition function

\newcommand{\autfromto}[3]{#1^{#2 \rightarrow #3}}
\newcommand{\autij}[3]{#1^{#2,#3}}

\newcommand{\concat}[0]{.}

\newcommand{\inportsets}{\mathcal{I}} % skupina mnozin vstupnich portu
\newcommand{\outportsets}{\mathcal{F}} % skupina mnozin vystupnich portu 
\newcommand{\innerinportsets}{\inportsets^\mathsf{in}}
\newcommand{\inneroutportsets}{\outportsets^\mathsf{in}}

\newcommand{\gsym}{\Gamma} % gate symbols
\newcommand{\pre}[1]{#1_{\mathit{pre}}} % pomocny automat u gate komplementace
\newcommand{\suf}[1]{#1_{\mathit{suf}}} % pomocny automat u gate komplementace
\newcommand{\insymgrp}[2]{\mathit{In}_{#1}^{#2}} % skupinky vstupnich portu podle gate symbolu
\newcommand{\outsymgrp}[2]{\mathit{Out}_{#1}^{#2}} % skupinky vystupnich portu podle gate symbolu

\newcommand{\detsuccsname}[0]{\mathit{powSC}}
\newcommand{\detsuccs}[1]{\detsuccsname(#1)}
\newcommand{\methodEqual}[0]{\texttt{equal}\xspace}
\newcommand{\methodDisjoint}[0]{\texttt{disjoint}\xspace}

\newcommand{\blinded}[1]{\ifx\blindreview\undefined #1 \else \textcolor{black!65}{[blinded for review]}\fi}

\newcommand{\aligater}{\ifx\blindreview\undefined\texttt{AliGater}\else\texttt{Tool}\fi\xspace}
\newcommand{\mata}{\texttt{Mata}\xspace}
\newcommand{\networkx}{\texttt{NetworkX}\xspace}
\newcommand{\rabit}{\texttt{Reduce}\xspace}
\newcommand{\NfaBench}[0]{\texttt{nfa-bench}\xspace}

\newcommand{\fwdpws}[0]{\texttt{fwd powerset}\xspace}
\newcommand{\revpws}[0]{\texttt{rev powerset}\xspace}
\newcommand{\plusmin}[0]{\texttt{+\;min}\xspace}
\newcommand{\sequential}[0]{\texttt{seq}\xspace}
\newcommand{\gate}[0]{\texttt{gate}\xspace}

\newcommand{\adela}[1]{\textcolor{orange}{\ifmmode \text{[A: #1]}\else [A: #1] \fi}}
\newcommand{\ol}[1]{\textcolor{blue}{\ifmmode \text{[OL: #1]}\else [OL: #1] \fi}}
\newcommand{\js}[1]{\textcolor{purple}{\ifmmode \text{[JS: #1]}\else [JS: #1] \fi}}

\newcommand{\duri}[1]{\textcolor{magenta}{\ifmmode \text{[Duri: #1]}\else [Duri: #1] \fi}}
\newcommand{\lh}[1]{\textcolor{violet}{\ifmmode \text{[LH: #1]}\else [LH: #1] \fi}}

\InputIfFileExists{cutmargins.tex}{%
	\pdfpagesattr{/CropBox [74 82 543 758]} % LNCS page made big
}{}

\InputIfFileExists{externalize-tikz.tex}{%
  \usepgfplotslibrary{external}
  \tikzexternalize[prefix=TikzPictures/]
}{}

\title{On Complementation of Nondeterministic Finite Automata without Full Determinization\ifTR{}\\ (Technical Report)\fi\\\vspace*{-3mm}
}
\author{Lukáš Holík\inst{1,2}\orcidID{https://orcid.org/0000-0001-6957-1651}
	\and Ondřej Lengál\inst{1}\orcidID{https://orcid.org/0000-0002-3038-5875}
	\and \\
  Juraj Major\inst{3}\orcidID{https://orcid.org/0009-0007-1871-9047}
	\and Adéla Štěpková\inst{3}\orcidID{https://orcid.org/0009-0004-4639-1673}
	\and Jan Strejček\inst{3}\orcidID{https://orcid.org/0000-0001-5873-403X}
}
\institute{Brno University of Technology, Brno, Czech Republic %\email{\{holik,lengal\}@fit.vutbr.cz}
  \and Aalborg University, Aalborg, Denmark
	\and Masaryk University, Brno, Czech Republic %\email{\{major,strejcek\}@fi.muni.cz, stepkova@mail.muni.cz}
}
\titlerunning{Complementation of Nondeterministic Finite Automata without Full Determinization}
\begin{document}

\maketitle

%{{{ abstract

\begin{abstract}
  \vspace{-6mm}
  Complementation of finite automata is a basic operation used in
  numerous applications. The standard way to complement a
  nondeterministic finite automaton (NFA) is to transform it into an
  equivalent deterministic finite automaton (DFA) and complement the
  DFA. The DFA can, however, be exponentially larger than the corresponding
  NFA.
  In this paper, we study several alternative approaches to complementation,
  which are based either on reverse powerset construction or on two novel
  constructions that exploit a~commonly occurring structure of NFAs.
  % Moreover, we propose a~heuristic for choosing a~complementation algorithm
  % based on the structure of the input NFA.
  Our experiment on a~large % and diverse
  data set shows that using
  a~different than the classical approach can, in many cases, yield significantly
  smaller complements.
  % A complementation of finite automata is a basic operation used in
  % numerous applications. The standard way to complement a
  % nondeterministic finite automaton (NFA) is to transform it into an
  % equivalent deterministic finite automaton (DFA) and complement the
  % DFA. The DFA can be exponentially larger than the corresponding
  % NFA. We study ways to complement NFAs without the determinization.
  % We suggest an approach based on co-determinization and several
  % complementation constructions for NFAs of specific shapes. All these
  % techniques are implemented and evaluated.
  %
  % \keywords{NFA \and complementation \and reverse powerset \and gate automata}
\end{abstract}

%}}}

%{{{ introduction

%!!!!!!!!!!!!!!!!!!!!!!!!!!!!!!!!
%\enlargethispage{2mm}
%!!!!!!!!!!!!!!!!!!!!!!!!!!!!!!!!

%%%%%%%%%%%%%%%%%%%%%%%%%%%%%%%%%%%%%%%%%%%%%%%%%%%%%%%%%%%%%%%%%%%%%%%%%%%%%%%%
\vspace{-7.5mm}
\section{Introduction}
\vspace{-1.5mm}
Complementation of \emph{finite automata} is an operation with many applications
in formal methods.
It is used, e.g.,
in regular model checking~\cite{KestenMMPS01,BouajjaniJNT00,BouajjaniHRV12},
representing extended regular expressions~\cite{VarataluVE25,ChocholatyHHHLS25},
to implement negation in automata-based decision procedures for logics such as
Presburger arithmetic~\cite{WolperB95,HabermehlHHHL24} or monadic
second order theories like WS1S or
MSO(Str)~\cite{Buchi60,ElgaardKM98,GlennG96,KelbMMG97,FiedorHLV19,FiedorHJLV17,AyariB00},
or as the basic underlying operation for testing language inclusion and
equivalence over automata.
Complementing \emph{deterministic finite automata} (DFAs) is
an easy task: it is sufficient to add a~single state, direct all missing
transitions to this state, and swap accepting and non-accepting states.
% a~linear-time\footnote{considering an efficient implementation of automata
% \ol{kick out?}}
% and constant-space operation %(one just needs
% (pass through the automaton
% to add missing transitions to a~new \emph{sink} state and then swap accepting and non-accepting states).
% \js{Prida to az linearne hodne hran, tak jaky constant-space?}

% Complementing \emph{nondeterministic} finite automata (NFAs) is, on the other
% hand, at most exponential and prone to explosion.

In practice, \emph{nondeterministic finite automata} (NFAs) are often favored over DFAs due to their potentially much
(up to exponentially) smaller size.
The classical approach to complementing NFAs goes through determinization of
the input NFA using the \emph{powerset construction} into a~DFA and then using DFA complementation.
While easy to implement, this approach is prone to cause a~blow-up in the number
of states, as determinizing an NFA with~$n$ states may, in the worst case,
result in a~DFA with~$2^n$ states~\cite{RabinS59} and the size of the
complement would then also be exponential.
%
% It seems that the powerset construction has been widely accepted as
% \emph{the} NFA complementation algorithm.
Some automata-based algorithms are highly sensitive to the sizes
of complement automata, such as decision procedures of certain
logics~\cite{WolperB95,HabermehlHHHL24,Buchi60,ElgaardKM98,GlennG96,KelbMMG97,FiedorHLV19,FiedorHJLV17,AyariB00},
where the output of the complement may be the basic structure over which another
complementation is performed (usually after projection, which can turn
a~potentially deterministic automaton into a~nondeterministic one);
for some of the logics, the increase in the size of the complement is
the underlying cause of their non-elementary complexity~\cite{Meyer72}.
%\js{tady chybi odkazy.}
Due to this, some of the applications tried to avoid complementation altogether,
for instance using symbolic
techniques~\cite{VarataluVE25,FiedorHLV19,FiedorHJLV17}.
This is, however, not always possible or feasible, as symbolic techniques often
disallow to use, e.g., standard automata reduction techniques
(cf.,~\cite{BustanG03,Hopcroft71}),
which often have a~great impact on the performance of the applications.

%\figPowersetEx  %%%%%%%%%%%%%%%%%%%%%

Since the lower bounds on the worst-case sizes of deterministic and complement
automata are both~$2^n$~\cite{SakodaS78,Birget93,HolzerK02,Jiraskova05}, one
might think that the determinization-based approach is optimal.
In practice, this is, however, far from the truth.
Consider, e.g., the NFA~$\Atwo$ accepting the language $\{a,b\}^*\concat \{a\}
\concat \{a,b\}^2$ shown in \cref{fig:reverse-ex} (left).
% , we note that
% the number~2 in the subscript of~$\Atwo$ corresponds to the repetition
% count~2 of the last part of the language description).
%We note that w
While the (minimal) deterministic complement 
%(shown as~$\DtwoCompl$ in \cref{fig:powerset-ex})
has 8~states, there exists an
NFA with 4 states given in \cref{fig:reverse-ex} (right) that is
a~complement of~$\Atwo$. 
% (we invite the reader to verify that the languages
% of~$\DtwoCompl$ and $\AtwoCompl$ are indeed equivalent).
This complement was obtained as follows.
%We obtained~$\AtwoCompl$ using the following procedure.
We reversed~$\Atwo$, then determinized it
(which
is easy as $\Atwo$ is reverse-deterministic), complemented the output by swapping
accepting and non-accepting states, reversed again, and, finally, removed
one unreachable state.
% Using an~intuitive notation that will be formally defined later, we can express
% the procedure as $\AtwoCompl = \rev{\co{\det{\rev{\Atwo}}}}$ (we assume the
% removal of unreachable states is implicit).
%
We can generalize the example above to the family of NFAs~$\A_n$ whose languages
are $\{a,b\}^*\concat \{a\} \concat \{a,b\}^n$ for $n \in \nat$.
Here, the size of the minimal complement DFA for the NFA~$\A_n$ is~$2^{n+1}$,
while the complement NFA constructed by the \emph{reverse powerset} procedure mentioned above has~$n+2$
states, the same as~$\A_n$.
This example is a~motivation for a deeper study of NFA
complementation.
% (instead of just blindly using the classical approach).

In this paper, we present several alternative approaches to NFA
complementation. Besides the reverse powerset complementation
(\cref{sec:reverse-powerset}), we introduce two novel complementation
constructions that target NFAs with a~particular structure (containing
several strongly connected components), common in practice. We first
introduce a basic version of the novel constructions on NFAs with a~very
restricted structure (\cref{sec:port-simple}) and then we briefly
present the generalized version of these constructions and the
generalized complementation problem that allows to combine these
constructions (\cref{sec:generalized-compl}).  Our experimental
evaluation (\cref{sec:experiments}) shows that in a~significant
number of cases, using an alternative complementation method can give
a~much smaller complement than the classical construction.
\ifTR
Additional details can be found in the appendices,
\else
Due to the page limit, many parts of the paper are relegated to the technical
report~\cite{techrep},
\fi
including a~precise description of the generalized complementation
constructions with the corresponding correctness proofs, an example of NFA subclass
with a~subexponential complement size, implementation details, and
additional experiments.
\vspace{-3.0mm}
\paragraph{Related Work.}\label{sec:label}
\vspace{-0.0mm}
%*******************************************************************************
While complementation of automata over infinite words is a~lively
topic (e.g., \cite{HavlenaLS25,AllredU18,BlahoudekDS20}), complementation of NFAs seems to be under-researched
% (especially from the practical side)
with not many relevant prior work.
The powerset approach to determinization, which is the basic
block of the classical complementation, can be traced to Rabin and
Scott~\cite{RabinS59}.
Optimizations of the powerset determinization
% using various relations (e.g., the simulation relation), which help
% mitigate the state space explosion,
were
proposed in~\cite{GlabbeekP08}.

From the theoretical side, the exponential lower bound of NFA complementation
was studied with respect to various alphabet sizes (the larger the alphabet
size, the easier it is to construct an NFA whose complement is forced to be exponential)
in~\cite{SakodaS78,Birget93,HolzerK02,Jiraskova05}.

% \lh{I think we actually don't need related work. There is only a very little of related work in the normal sense of the word. Rabin and Scot, that paper Five complementation Algorithms should be there, and the basic complexity results. 
% It can be merged in intro.} 
% \lh{Refs to buchi are optional, and if we talk about them, lets mention that they build on the classical powerset (advances there might lead to advances there).}

%}}}
%{{{ preliminaries

%%%%%%%%%%%%%%%%%%%%%%%%%%%%%%%%%%%%%%%%%%%%%%%%%%%%%%%%%%%%%%%%%%%%%%%%%%%%%%%%
%%%%%%%%%%%%%%%%%%%%%%%%%%%%%%%%%%%%%%%%%%%%%%%%%%%%%%%%%%%%%%%%%%%%%%%%%%%%%%%%
\vspace{-3.0mm}
\section{Preliminaries}\label{sec:prelim}
\vspace{-2.0mm}
%%%%%%%%%%%%%%%%%%%%%%%%%%%%%%%%%%%%%%%%%%%%%%%%%%%%%%%%%%%%%%%%%%%%%%%%%%%%%%%%

%This section recalls definitions and the classical determinization-based
%complementation. %\lh{or powerset-based? Classica?}

%\textbf{Formal Languages.}
An \emph{alphabet} $\Sigma$ is a finite nonempty set of
\emph{symbols}.  
A \emph{word} over %an \emph{alphabet}
$\Sigma$ %(a finite non-empty set of \emph{symbols})
is a %finite
sequence $u = u_1\ldots u_n$ where $u_i\in \Sigma$ for all $1\leq i \leq n$,  
%of symbols from $\Sigma$, 
with its
\emph{length} $n$ denoted by $|u|$.
The \emph{empty word} is denoted by $\varepsilon$.  
The set of all words over $\Sigma$ is denoted by $\Sigma^*$, and 
its subsets are \emph{languages} over $\Sigma$.
The \emph{concatenation} of $u$ with a word $v = u_{n+1} \ldots u_m$ is the word
$uv = u_1 \ldots u_m$. 
%The \emph{concatenation} of $u$ with a word $v$ is denoted by 
%$uv$. 
The \emph{reverse} of $u$ is the word $\rev{u} = u_nu_{n-1} \ldots u_1$. 
The \emph{concatenation} of languages $L, L' \subseteq \Sigma^*$ is the
language $L\concat L' = \{uv \mid u \in L, v \in L'\}$.
The~\emph{reverse} of $L$ is the language
$\rev{L} = \{ \rev{w} \mid w \in L\}$ and its 
\emph{complement} is the language
$\compl{L} = \Sigma^* \smallsetminus L$.

% \ol{the stacking of superscripts and overlines screws up line
% spacing in places. One might consider using sth like
% $L^\complement$ and $L^{\complement_\Lambda}$, what do you
% think?}\js{I suggest to use uniform notation for complement,
% reverse, determinization. So I've redefined the macros.}

\textbf{Finite Automata.}
A \emph{nondeterministic finite automaton (NFA)} is defined as a tuple
$\A=(Q, \Sigma, \delta, I, F)$, where $Q$ is a %~nonempty
finite set of \emph{states}, $\Sigma$ is an alphabet,
$\delta\subseteq Q \times \Sigma \times Q$ is a \emph{transition
  relation}, $I \subseteq Q$ is a set of \emph{initial states}, and
$F \subseteq Q$ is a set of \emph{accepting states} (also called
\emph{final states}). We write $q \xrightarrow{a} r \in \delta$
instead of $(q,a,r)\in\delta$. If $\delta$ is clear from the context,
we write just $q \xrightarrow{a} r$.
%%%  I guess we will not need this if we have the arrow notation
The \emph{size} of $\A$ is defined as $|\A| = |Q|$.
We abuse the notation and use~$\delta$ also as the function
$\delta\colon Q\times\Sigma\rightarrow 2^Q$ defined as
\mbox{$\delta(q,a)=\{r \in Q\mid q\xrightarrow{a} r\}$.}
We also extend the
function $\delta$ to sets of states $P\subseteq Q$ as
$\delta(P,a)=\bigcup_{q\in
  P}\delta(q,a)$. 
% Finally, we extend it also
% to words $w\in\Sigma^*$ using the base case $\delta(P,\varepsilon)=P$
% and the inductive step $\delta(P,va)=\delta(\delta(P,v),a)$ for any
% $P\subseteq Q$, $v\in\Sigma^*$, and $a\in\Sigma$. 
%
A~\emph{deterministic finite automaton (DFA)} is an NFA
% $(Q, \Sigma, \delta, I, F)$
with $|I| = 1$ and
$|\delta(q,a)|\le 1$ for all $q \in Q$ and $a \in \Sigma$. A~DFA is
\emph{complete} if $|\delta(q,a)|=1$ for all $q \in Q$ and~$a \in \Sigma$.
A \emph{run} of an NFA $\A$ over a~word $v = v_1v_2 \ldots v_n$ is a
sequence of transitions $q_0\xrightarrow{v_1} q_{1}, q_1\xrightarrow{v_2} q_{2},\ldots,q_{n-1}\xrightarrow{v_n} q_{n}$ with $q_0\in I$.
It is \emph{accepting} if $q_n \in F$.
A state that appears in a run over some word~$v \in \Sigma^*$ is
\emph{reachable}, else it is \emph{unreachable}. 
A state is \emph{reachable from a state} $q$ if it is reachable in the NFA $(Q, \Sigma, \delta, \{q\}, F)$.
%
%$ states $\sigma = q_0q_1\ldots q_n$, where $q_0 \in I$ and,
%for each $1 \leq i \leq n$, we have
%$q_{i-1} \xrightarrow{v_i} q_i \in \delta$. 
%The run $\sigma$ is
%\emph{accepting} if $q_n \in F$. 
$\A$ \emph{accepts} the
language $\langof{\A}$ of all words over $\Sigma$ for which it has an accepting run. 
%
%$\langof{\A} = \{w \in \Sigma^{*} \mid {}$ there is an accepting
%  run of $\A\textrm{ over }w\}$. 
Automata $\A$ and $\B$ are
\emph{equivalent} iff $\langof{\A} = \langof{\B}$.

The \emph{reverse} of $\A$ is the NFA
$\rev{\A} = (Q, \Sigma, \rev{\delta}, F, I)$, where
$ \rev{\delta} = \{ r \xrightarrow{a} q\mid q \xrightarrow{a} r \in
\delta \}$. $\A$ is called \emph{reverse-deterministic} iff $\rev{\A}$ is 
a~DFA.
An NFA $\C$ is called a~\emph{complement} of $\A$ with respect to an alphabet $\Lambda$ if
$\langof{\C}=\Lambda^*\smallsetminus\langof{\A}$.
% We use $\complwith{\A}{\Lambda}$ to denote such a complement.
If $\Lambda$ is not specified, we assume that $\Lambda = \Sigma$.
% and use the notation $\compl{\A}$.
% \js{I've cancelled the notation $\compl{\A}$ because it looks as a function but it is not.}
%
Given two NFAs $\A_1 = (Q_1, \Sigma, \delta_1, I_1, F_1)$ and
$\A_2 = (Q_2, \Sigma, \delta_2, I_2, F_2)$ where
$Q_1 \cap Q_2 = \emptyset$, their \emph{union} is the NFA
$\A_1 \Cup \A_2 = (Q_1 \cup Q_2, \Sigma, \delta_1 \cup \delta_2, I_1
\cup I_2, F_1 \cup F_2)$. 
% Their \emph{intersection} is the NFA
% $\A_1 \Cap \A_2 = (Q_1 \times Q_2, \Sigma, \delta, I_1 \times I_2, F_1
% \times F_2)$ where
% $\delta=\{(q_1, q_2)\xrightarrow{a}(r_1,r_2) \mid
% q_1\xrightarrow{a}r_1\in\delta_1 \textrm{ and }
% q_2\xrightarrow{a}r_2\in\delta_2\}$.
%
% A~state $q$ is \emph{dead} if no final state is reachable from $q$.
%
A \emph{strongly connected component (SCC)} of $\A$ is a~maximal subset
$C \subseteq Q$ in which every state is reachable from every state.  
%%We always refer to strongly connected components
%as SCCs, while 
%
Note that the term \emph{component} will be used in a more general
sense later.

\textbf{Forward Powerset Complementation.}
\label{sec:forward-powerset} 
The standard complementation first uses the \emph{powerset
  construction} to transform a given NFA
$\A = (Q, \Sigma, \delta, I, F)$ into an equivalent complete DFA
$\det{\A} = (Q', \Sigma, \delta', I', F')$, where $Q'=2^Q$,
$\delta' = \{P\xrightarrow{a} \delta(P,a) \mid P\in Q', a\in\Sigma\}$,
$I'=\{I\}$, and $F' = \{ P \in Q' \mid P \cap F \neq \emptyset\}$.  In
the following, we assume that $\det{\A}$ does not contain unreachable
states (only the reachable part of $Q'$ is constructed). This does not
improve the upper bound on the size of the DFA, which is still
$2^{|Q|}$.
%
% where $Q' \subseteq 2^Q$ and $\delta'$ are the minimal sets that can
% accommodate $I' = \{I\}$ and $P\xrightarrow{a} \delta(P,a)$ for each
% $P\in Q', a \in
% \Sigma$; and $F' = \{ P \in Q' \mid P \cap F \neq \emptyset\}$.
%
%Here, $Q' = 2^Q$, $I' = \{I\}$,
%$F' = \{ P \in Q' \mid P \cap F \neq \emptyset \}$, and
%$\delta' = \{P\xrightarrow{a} \delta(P,a) \mid P\in Q', a \in
%\Sigma\}$, wheIn the following text,
%we assume that~$\det{\A}$ does not contain unreachable states (as if produced by the standard fixpoint construction). 
%, i.e., that it is
%constructed by a~fixpoint procedure adding successors of already
%computed states starting with $\{I\}$.
% Note that many states in~$Q'$ may be
% unreachable from the only initial state $I$ and thus
% irrelevant.
% We assume that these states are removed from the
% automaton. In practice, these states are never added as~$Q'$ is constructed by
% iteratively adding successors of already
% computed states starting with $I$. The removal of unreachable states
%This, however,
%does not improve the upper bound on the size of the DFA, which is
%still~$2^{|Q|}$.
% \lh{maybe we could describe the construction that explores the reachable space only right away and save space, some succinct fixpoint stuff perhaps}
% \js{You can try. But I guess that the benefit is not big enough.}
%
Given a complete DFA $\D=(Q, \Sigma, \delta, I, F)$, its complement can be constructed as 
$\co{\D}=(Q, \Sigma, \delta, I, Q \smallsetminus F)$.
A complement of an NFA~$\A$ can thus be constructed as $\co{\det{\A}}$.
We call this construction \emph{forward powerset complementation}.
%
% Since a complete DFA $\D=(Q, \Sigma, \delta, I, F)$ is 
% complemented as $\co{\D}=(Q, \Sigma, \delta, I, Q \smallsetminus F)$,
% an NFA~$\A$ is complemented as $\co{\det{\A}}$.
% We call this construction \emph{forward powerset complementation}.
%
% to distinguish it from the \emph{reverse powerset complementation}
% presented below.

%}}}
%{{{ reverse powerset

%%%%%%%%%%%%%%%%%%%%%%%%%%%%%%%%%%%%%%%%%%%%%%%%%%%
\newcommand{\figReverse}[0]{
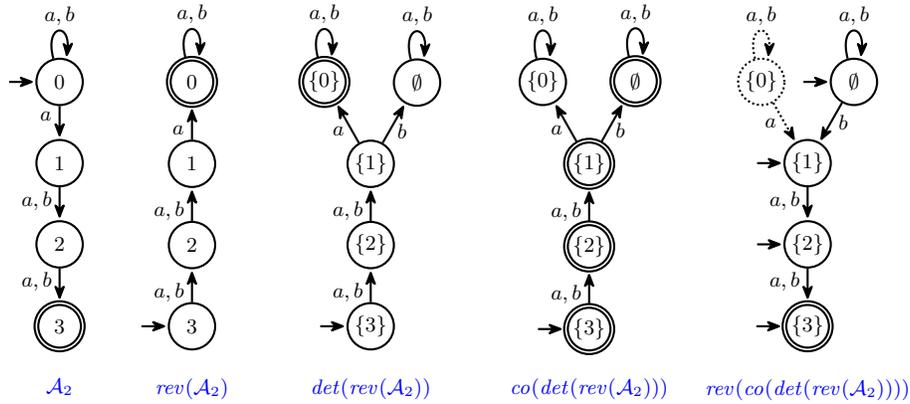
\begin{figure}[tb]
  \centering
  \begin{tikzpicture}[scale=.88,smallautomaton,node distance=5mm,every state/.style={minimum size=7mm,inner sep=0pt}]
    \tikzset{aname/.style={blue,yshift=2mm}}
    \begin{scope}
      \node[state, initial] (0) {$0$};
      \node[state, below = of 0] (1) {$1$};
      \node[state, below = of 1] (2) {$2$};
      \node[state, below = of 2, accepting] (3) {$3$};
      \node[aname, below = of 3] (label) {$\A_2$}; 
        
      \path[->]
      (0) edge [loop above] node[above=-1pt] {$a, b$} ()
      (0) edge node[left, pos=0.3] {$a$} (1)
      (1) edge node[left, pos=0.4] {$a,b$} (2)
      (2) edge node[left, pos=0.4] {$a,b$} (3)
      ;
    \end{scope}
    \begin{scope}[xshift=2cm]
      \node[state, accepting] (0) {$0$};
      \node[state, below = of 0] (1) {$1$};
      \node[state, below = of 1] (2) {$2$};
      \node[state, below = of 2, initial] (3) {$3$};
      \node[aname, below = of 3] (label) {$\rev{\A_2}$}; 
        
      \path[->]
      (0) edge [loop above] node[above=-1pt] {$a, b$} ()
      (1) edge node[left, pos=0.3] {$a$} (0)
      (2) edge node[left, pos=0.3] {$a,b$} (1)
      (3) edge node[left, pos=0.3] {$a,b$} (2)
      ;
    \end{scope}
    \begin{scope}[xshift=4cm]
      \node[state, accepting] (0) {$\{0\}$};
      \node[state, below = of 0,xshift=7mm] (1) {$\{1\}$};
      \node[state, above = of 1, xshift=7mm] (0b) {$\emptyset$};
      \node[state, below = of 1] (2) {$\{2\}$};
      \node[state, below = of 2, initial] (3) {$\{3\}$};
      \node[aname, below = of 3] (label) {$\det{\rev{\A_2}}$}; 
        
      \path[->]
      (0) edge [loop above] node[above=-1pt] {$a, b$} ()
      (0b) edge [loop above] node[above=-1pt] {$a, b$} ()
      (1) edge node[left, pos=0.3] {$a$} (0)
      (1) edge node[right, pos=0.3] {$b$} (0b)
      (2) edge node[left, pos=0.3] {$a,b$} (1)
      (3) edge node[left, pos=0.3] {$a,b$} (2)
      ;
    \end{scope}
    \begin{scope}[xshift=7.3cm]
      \node[state] (0) {$\{0\}$};
      \node[state, below = of 0, accepting, xshift=7mm] (1) {$\{1\}$};
      \node[state, above = of 1, accepting, xshift=7mm] (0b) {$\emptyset$};
      \node[state, below = of 1, accepting] (2) {$\{2\}$};
      \node[state, below = of 2, accepting, initial] (3) {$\{3\}$};
      \node[aname, below = of 3,yshift=.5mm] (label) {$\co{\det{\rev{\A_2}}}$}; 
        
      \path[->]
      (0) edge [loop above] node[above=-1pt] {$a, b$} ()
      (0b) edge [loop above] node[above=-1pt] {$a, b$} ()
      (1) edge node[left, pos=0.3] {$a$} (0)
      (1) edge node[right, pos=0.3] {$b$} (0b)
      (2) edge node[left, pos=0.3] {$a,b$} (1)
      (3) edge node[left, pos=0.3] {$a,b$} (2)
      ;
    \end{scope}
    \begin{scope}[xshift=10.6cm,yshift=0mm]
      \node[state, densely dotted] (0) {$\{0\}$};
      \node[state, below = of 0, initial, xshift=7mm] (1) {$\{1\}$};
      \node[state, above = of 1, initial, xshift=7mm] (0b) {$\emptyset$};
      \node[state, below = of 1, initial] (2) {$\{2\}$};
      \node[state, below = of 2, accepting, initial] (3) {$\{3\}$};
      \node[aname, below = of 3] (label) {$\rev{\co{\det{\rev{\A_2}}}}$}; 
        
      \path[->]
      (0) edge [loop above, densely dotted] node[above=-1pt] {$a, b$} ()
      (0b) edge [loop above] node[above=-1pt] {$a, b$} ()
      (0) edge [densely dotted] node[left] {$a$} (1)
      (0b) edge node[right] {$b$} (1)
      (1) edge node[left, pos=0.4] {$a,b$} (2)
      (2) edge node[left, pos=0.4] {$a,b$} (3)
      ;
    \end{scope}
  \end{tikzpicture}
  \caption{NFA $\A_2$ accepting language
    $\{a, b\}^*.\{a\}.\{a, b\}^2$ and all phases of its reverse
      powerset complementation; the dotted part of the complement
      automaton is unreachable.}
  \label{fig:reverse-ex}
  \vspace*{-4mm}
\end{figure}
}

%%%%%%%%%%%%%%%%%%%%%%%%%%%%%%%%%%%%%%%%%%%%%%%%%%%%%%%%%%%%%%%%%%%%%%%%%%%%%%%%
\vspace{-3.0mm}
\section{Reverse Powerset Complementation} \label{sec:reverse-powerset}
\vspace{-2.0mm}
%%%%%%%%%%%%%%%%%%%%%%%%%%%%%%%%%%%%%%%%%%%%%%%%%%%%%%%%%%%%%%%%%%%%%%%%%%%%%%%%

The idea of this approach to complementation is simple: if we reverse an
automaton, then complement it (for example using the powerset
complementation given above) and reverse again, we obtain the complement of
the original automaton. Formally, given an NFA $\A$, its complement can
be constructed as $\rev{\co{\det{\rev{\A}}}}$. 
% Optionally, unreachable states may be removed in the final automaton. 
We call this approach 
\emph{reverse powerset complementation}. \cref{fig:reverse-ex} shows
all phases of this complementation process on an example automaton.
% \lh{I would be careful not to sound like that we claim this as our new invention (though we can claim the observation and the experimental evidence that it works well)}
% \js{We do not claim here to be the authors of reverse powerset. Maybe we indicate it a bit in the Introduction.} 
%\JS{As far as we know, this complementation was never presented before.}

\figReverse   %%%%%%%%%%%%%%

Contrary to the forward powerset construction, this method can yield
a~(forward-) nondeterministic automaton, which may be significantly smaller
than the minimal deterministic complement. This can be documented by a
generalization of the automata in \cref{fig:reverse-ex}. For any
$n \in \mathbb{N}$, there exists an NFA $\A_n$ with $n+2$ states that
accepts the language $\{a, b\}^*.\{a\}.\{a, b\}^n$. While the minimal forward powerset complement
of $\A_n$ has $2^{n+1}$ states % This result is folklore
  % (\lh{I would not say, in our group only});
(intuitively, the DFA must store a bit vector of $n+1$ elements
tracking which of the last $n+1$ symbols read were $a$), the reverse
powerset complementation produces an NFA with $n+2$ reachable
states. This is due to $\A_n$ being reverse-deterministic. Our
experiments in \ifTR\cref{sec:additional-exp}\else\cite{techrep}\fi{} show that the reverse
powerset construction outperforms the forward one in many cases.

\textbf{Heuristic for Forward vs.\ Reverse Powerset.}
Naturally, the powerset construction can be efficient
in one direction (forward or reverse) while causing a blow-up in the
other. To address this, we designed a cheap heuristic to choose the more
favorable direction for a given NFA, for cases when running
a~portfolio is too expensive.

% Intuitively,
% the heuristic computes the degree of nondeterminism in an NFA, which is computed
% using the sizes of different sets of successors of individual states.
The heuristic sums the sizes of all powerset successors of each state
in a given NFA.  Formally, for an NFA
$\A = (Q, \Sigma, \delta, I, F)$, we define
$\mathit{sc}(q) = \{\delta(q, a) \mid a \in \Sigma \}$ for each state
$q \in Q$ and compute $
\detsuccs{\A} = |I| + \sum_{q \in Q} \sum_{ S \in \mathit{sc}(q)} |S|$.
We emphasize that if $q$ has the same successors under two different
symbols $a, b \in \Sigma$, then the set $\delta(q, a) = \delta(q, b)$
contributes to the sum only once.

Intuitively, a higher value of $\detsuccs{\A}$ should indicate a
higher number of states produced by the powerset construction applied
to $\A$. Hence, comparing $\detsuccs{\A}$ and $\detsuccs{\rev{\A}}$
gives a hint which of the powerset
constructions has a greater risk of blow-up. The heuristic's performance is
experimentally evaluated in \ifTR\cref{sec:additional-exp}\else\cite{techrep}\fi.

%!!!!!!!!!!!!!!!!!!!!!!!!!!!!!!!!
%\enlargethispage{2mm}
%!!!!!!!!!!!!!!!!!!!!!!!!!!!!!!!!

%}}}
%{{{ sequential and gate complementation

%%%%%%%%%%%%%%%%%%%%%%%%%%%%%%%%%%%%%%%%%%%%%%%%%%%%%%%%%%%%%%%%%%%%%%%%%%%%%%%%
%%%%%%%%%%%%%%%%%%%%%%%%%%%%%%%%%%%%%%%%%%%%%%%%%%%%%%%%%%%%%%%%%%%%%%%%%%%%%%%%
\vspace{-3.0mm}
\section{Sequential and Gate Complementation} \label{sec:port-simple}
\vspace{-2.0mm}
%%%%%%%%%%%%%%%%%%%%%%%%%%%%%%%%%%%%%%%%%%%%%%%%%%%%%%%%%%%%%%%%%%%%%%%%%%%%%%%%

This section introduces the basic ideas of two techniques called
\emph{sequential} and \emph{gate complementation} with use on
automata with a specific shape. The general form of these techniques
is presented in \cref{sec:generalized-compl}.

Consider an NFA $\A$ that can be seen as two disjoint NFAs
$\A_1=(Q_1,\Sigma,\delta_1,I_1,\{q_F\})$ and
$\A_2=(Q_2,\Sigma,\delta_2,\{q_I\},F_2)$ connected with a single
transition $q_F\xrightarrow{c}q_I$ for some $c\in\Sigma$, i.e.,
$\A=(Q_1\cup
Q_2,\Sigma,\delta_1\cup\delta_2\cup\{q_F\xrightarrow{c}q_I\},I_1,F_2)$.
The automaton $\A$ accepts the language $L_1.\{c\}.L_2$ where
$L_1=\langof{\A_1}$ and $L_2=\langof{\A_2}$. An example of such an NFA
$\A$ and the corresponding automata~$\A_1$ and~$\A_2$ can be found in
\cref{fig:port-simple-ex}. The transition $q_F\xrightarrow{c}q_I$ is
called the \emph{transfer transition} and automata $\A_1$ and $\A_2$ are
referred to as the \emph{front} and the \emph{rear component} of $\A$,
respectively.
% \lh{can we rename $c$ to $a$, as on the
%   figure?}\js{$c$ is used in the second subsection where you do not want to use $a$.}

%*******************************************************************************
\vspace{-3.0mm}
\subsection{Sequential Complementation}\label{sec:seq-simple}
\vspace{-1.5mm}
%*******************************************************************************
\emph{Sequential complementation} builds a potentially nondeterministic
complement $\C$ of $\A$ from a complete DFA $\A_1'=\det{\A_1}$
% \lh{we could maybe use $\bar\A$ or $\overline \A$ as a default, instead of $det(A)$, avoid this renaming and make it shorter? Or $\A^{\mathtt{pow}}, \A^{\mathtt{p}}$ maybe}
% \js{I think that it woud be nonsystematic, as we have several functions on automata, not only $\det{}$.}
equivalent to $\A_1$ and an NFA complement $\C_2$ of $\A_2$. The
technique is called sequential complementation because we first
complement $\A_2$ using an arbitrary complementation approach,
and only then we \mbox{build the complement of~$\A$.}

The construction is based on the following observation. The language
$\co{L_1.\{c\}.L_2}$ consists of words $w$ such that, for all pairs of
words $u, v$ satisfying $ucv = w$, if $u \in L_1$ then
$v \in \compl{L_2}$.
Note that if $w$ does not contain any $c$, the
condition is trivially satisfied.
Therefore, $\C$ will simulate $\A_1'$ and whenever a final state of
$\A_1'$ is visited, that is, $\C$ finished reading a % possible
prefix
$u \in L_1$, a~transition under~$c$ initiates a new instance of $\C_2$ that will check
that the corresponding suffix $v$ is indeed in $\co{L_2}$. Automaton $\C$ will
accept if each initiated instance of $\C_2$ accepts, meaning that for
all possible splittings of the input word $w$ into $ucv$ where
$u \in L_1$, it holds that $v \in \co{L_2}$.
% (note that this also includes the case when there is no instance of
% $\C_2$ yet, meaning that no prefix belonging to $L_1$ was read so
% far).

Formally, let $\A_1' = (Q'_1, \Sigma, \delta'_1, \{q'_0\}, F'_1)$ be a complete DFA
representing the language~$L_1$ and
$\C_2=(\elc{Q}_2, \Sigma, \elc{\delta}_2, \elc{I}_2, \elc{F}_2)$ be an
NFA representing $\co{L_2}$. 
% \lh{the $\A'_1$  $\C_2$ above} \js{solved?}
We construct an NFA
$\C = (\elc{Q}, \Sigma, \elc{\delta}, \elc{I}, \elc{F})$ representing
$\co{L_1.\{c\}.L_2}$ as follows.
\begin{itemize}
\item $\C$'s states are pairs composed of the current state of $\A'_1$ and states of instances of $\C_2$, $\elc{Q}= Q'_1 \times 2^{\elc{Q}_2}$.
\item For each $a\in\Sigma$, the transition relation $\elc{\delta}$
  simulates the corresponding transition of $\A_1'$ and arbitrary
  transitions under $a$ of the running instances of $\C_2$. Moreover,
  it initiates a~new instance of $\C_2$ whenever $\A_1'$ moves from an
  accepting state by reading~$c$. Formally, for each
  $(p,\{r_1,\ldots,r_n\})\in\elc{Q}$, $a\in\Sigma$, and transitions
  $p\xrightarrow{a}q\in\delta'_1$,
  $r_i\xrightarrow{a}s_i\in\elc{\delta}_2$ for all $1\le i\le n$, the
  transition relation $\elc{\delta}$ contains transitions
  \begin{itemize}
  \item $(p,\{r_1,\ldots,r_n\})\xrightarrow{a}(q,\{s_1,\ldots,s_n\})$
    if $p\not\in F'_1$ or $a\neq c$, and
  \item
    $(p,\{r_1,\ldots,r_n\})\xrightarrow{a}(q,\{s_1,\ldots,s_n\}\!\cup\!\{s_0\})$
    for all $s_0\!\in\!\elc{I}_2$ if $p\!\in\! F'_1$ and $a\!=\!c$.
  \end{itemize}
  % Note that if at least one $r_i$ has no successors under $a$, the
  % whole state $(p,\{r_1,\ldots,r_n\})$ has no successors under $a$.
  % intuitively, the runs of one copy of $\C_2$ over some suffix $v$
  % represented by $r_i$ are definitely not accepting, so there is no
  % point in continuing with this path. Also, $|\{s_1, \ldots, s_n\}|$
  % can be smaller than $n$ since any two $s_i$s can be identical. In
  % this case, the runs of two (or more) simulations of $\C_2$ meet in
  % the same state and are indistinguishable from now on.  
\item $\C$ starts in the initial state $q'_0$ of $\A_1'$ with no
  running instance of $\C_2$, i.e., $\elc{I}=\{(q'_0,\emptyset)\}$.
  % unless $q_0$ is accepting state of
  % $\A_1'$. If $q_0$ is accepting, we initiate one instance of $\C_2$ in
  % its arbitrary initial state. Formally, $\elc{I}=\{(q_0,\emptyset)\}$
  % if $q_0\notin F_1$, otherwise
  % $\elc{I}=\{(q_0, \{s_0\}) \mid s_0\in \elc{I}_2\}$.
\item $\C$ accepts whenever all running instances of $\C_2$ accept, i.e., $\elc{F}=Q'_1 \times 2^{\elc{F}_2}$.
\end{itemize}
\vspace{-1mm}

\begin{theorem}
  The NFA $\C$ accepts $\compl{\langof{\A}}$. %\compl{L_1.\{c\}.L_2}=$.
\end{theorem}
\begin{proof}[Proof (sketch)]
  Recall that $\langof{\A}=L_1.\{c\}.L_2$.  First, consider $w \in
  \langof{\A}$. Then there are $u \in L_1=\langof{\A_1'}$ and $v \in
  L_2$ such that $w = ucv$. As
  $\A_1'$ is deterministic, it has to reach a~state $q_f\in
  F_1'$ after reading $u$. Hence,
  $\C$ can reach only states of the form $(q_f,R)$ after reading
  $u$. When $\C$ reads
  $c$ from this state, it reaches a state $(q',R')$, where
  $R'$ has to contain some initial state $s_0$ of
  $\C_2$. However, $v\in L_2$ implies that $\C_2$ does not accept~$v$.
  Hence, each state $(q'',R'')$ of $\C$ reached from
  % $(q'',R'')\in\elc{\delta}((q',R'),v)$ reached from
  $(q',R')$ by reading
  $v$ is not accepting as it cannot satisfy
  $R''\subseteq\elc{F}_2$. To sum up, $\C$ has no accepting run over $ucv=w$.
  
  Now assume that $w \notin \langof{\A}$. As
  $\A_1'$ is deterministic and complete, it has a single run over
  $w$. Whenever the run reaches an accepting state over some prefix
  $u$ of $w$, we know that $u\in
  L_1$ and thus the corresponding suffix cannot be of the form
  $cv\in\{c\}.L_2$ as that would contradict the assumption
  $w=ucv\not\in\langof{\A}$. In other words, if the prefix
  $u$ is followed by $c$, then
  $\C_2$ has an accepting run over the corresponding suffix
  $v$ as $v\not\in L_2$. We can construct an accepting run of
  $\C$ over $w$ such that whenever the automaton
  $\A_1'$ tracked in the first element of the states of
  $\C$ reaches an accepting state and $\C$ reads
  $c$, we add to the second element of the state of
  $\C$ the initial state of the corresponding accepting run of~$\C_2$
  and then follow this run in the future transitions of
  $\C$. After reading the whole~$w$,
  the second element of the reached state of
  $\C$ will contain only accepting states of
  $\elc{F}_2$. Thus, the constructed run of $\C$ over
  $w$ is accepting and $w\in\langof{\C}$.
\end{proof}

%The correctness of the construction is proven in \ifTR\cref{app:proofs}\else\cite{techrep}\fi.

\begin{figure}[tp]
    \centering
    \begin{tikzpicture}[scale=.88,smallautomaton,node distance=5mm,every state/.style={minimum size=6mm,inner sep=1pt}]
    \tikzset{aname/.style={blue}}
    \begin{scope}
        \node[state, initial, initial where = above] (0) {$0$};
        \node[state, below = of 0] (1) {$1$};
        \node[state, below = of 1,yshift=-4mm] (2) {$2$};
        \node[state, below = of 2] (3) {$3$};
        \node[state, accepting, below = of 3] (4) {$4$};
         
        \path[->]
        (0) edge node[left] (l1) {$a, b$} (1)
        (1) edge node[left] {$a$} (2)
        (2) edge [loop left] node {$a, b$} (2)
        (2) edge node[left] {$a$} (3)
        (3) edge node[left] {$a, b$} (4)
        ;

        \node[aname, left = of l1,xshift=5mm] (a1) {$\A_1$};
        \node[aname] (a2) at (a1|-3) {$\A_2$};
        \node[aname, left = of 0,yshift=6mm,xshift=7mm] (a) {$\A$};
        \begin{scope}[on background layer]
          \node [draw=blue!20, fill=blue!20, fit={(0) (1) (a1)},
            rounded corners=14pt, inner sep=3pt] (b1) {};
          \node [draw=blue!20, fill=blue!20, fit={(2) (4) (a2)},
          rounded corners=14pt, inner sep=3pt] (b2) {};
        \end{scope}
    \end{scope}
        
    \begin{scope}[xshift=3cm]
        \node[state, initial, initial where = above] (0) {$0$};
        \node[state, accepting, below = of 0] (1) {$1$};
        \node[aname, left = of 0,yshift=6mm,xshift=7mm] (a1) {$\A_1$};

        \node[state, initial, initial where = above, below = of 1,yshift=-4mm] (2) {$2$};
        \node[state, below = of 2] (3) {$3$};
        \node[state, accepting, below = of 3] (4) {$4$};
        \node[aname, left = of 2,yshift=6mm,xshift=7mm] (a2) {$\A_2$};

        \path[->]
        (0) edge node[left] {$a, b$} (1)
        ;

        \path[->]
        (2) edge [loop left] node {$a, b$} (2)
        (2) edge node[left] {$a$} (3)
        (3) edge node[left] {$a, b$} (4)
        ;
    \end{scope}

    \begin{scope}[xshift=6.4cm]
        \node[state, initial, initial where = above] (0) {$0$};
        \node[state, accepting, below = of 0] (1) {$1$};
        \node[state, right = of 1,xshift=1mm] (s) {$s$};
        \node[aname, left = of 0,yshift=6mm,xshift=7mm,anchor=east] (a1) {$\det{\A_1}{=}\A_1'$}; 
        
        \node[state, initial, initial where = above, below = of 1,yshift=-4mm] (2) {$5$};
        \node[state, initial, below = of 2] (3) {$6$};
        \node[state, initial, accepting, below = of 3] (4) {$7$};
        \node[aname, left = of 2,yshift=6mm,xshift=7mm] (c2) {$\C_2$};

        \path[->]
        (0) edge node[left] {$a, b$} (1)
        (1) edge node[above] {$a, b$} (s)
        (s) edge [loop above] node {$a, b$} (s)
        ;

        \path[->]
        (2) edge [loop left] node {$a, b$} (2)
        (2) edge node[left] {$b$} (3)
        (3) edge node[left] {$a, b$} (4)
        ;
    \end{scope}
    
    \tikzset{wstate/.style={state,rounded rectangle,text width=9mm,inner sep=1pt,align=center}}
    
    \begin{scope}[xshift=10.2cm]
        \node[wstate, accepting, initial, initial where = above] (00) {$0,\emptyset$};
        \node[wstate, accepting, below = of 00] (10) {$1,\emptyset$};
        \node[wstate, accepting, right = of 10] (s0) {$s, \emptyset$};
        \node[wstate, below = of 10,yshift=-4mm] (26) {$s, \{5\}$};
        \node[wstate, below = of 26] (27) {$s, \{6\}$};
        \node[wstate, below = of 27, accepting] (28) {$s, \{7\}$};
        \node[aname, left = of 00,yshift=6mm,xshift=9mm] (c) {$\C$}; 
         
        \path[->]
        (00) edge node[left] {$a, b$} (10)
        (10) edge [left] node[above, pos=0.5] {$b$} (s0)
        (s0) edge [loop above] node {$a, b$} (s0)
        (10) edge node[left] {$a$} (26)
        (10) edge[out=-140,in=140] node[left,pos=.7] {$a$} (27)
        (10) edge[out=-150,in=150] node[left,pos=.7,overlay] {$a$} (28)
        (26) edge node[right] {$b$} (27)
        (26) edge [loop right,looseness=5] node {$a, b$} (26)
        (27) edge node[right] {$a, b$} (28)
        ;
    \end{scope}
    \end{tikzpicture}
     
    \caption{An NFA $\A$, its front and rear components $\A_1,\A_2$,
      a~complete DFA $\A_1'$ equivalent to~$\A_1$, a~complement $\C_2$
      of $\A_2$, and the complement $\C$ of $\A$ constructed from
      $\A_1'$ and $\C_2$.}
    % An NFA $\A$ the language $L^n_1 \cdot L^n_2$ (a), a complete DFA
    % $\A^n_1$ accepting $L^n_1$ (b), an NFA $\C^n_2$ accepting
    % $\compl{L^n_2}$ (c), and an NFA $\C^n$ accepting
    % $\compl{L^n_1 \cdot L^n_2}$ (d), where
    % $L^n_1 = \{a, b\}^n \cdot \{a\}$,
    % $L^n_2 = \{a, b\}^* \cdot \{a\} \cdot \{a, b\}^n$, and $n = 1$
    \label{fig:port-simple-ex}
    \vspace*{-4mm}
  \end{figure}
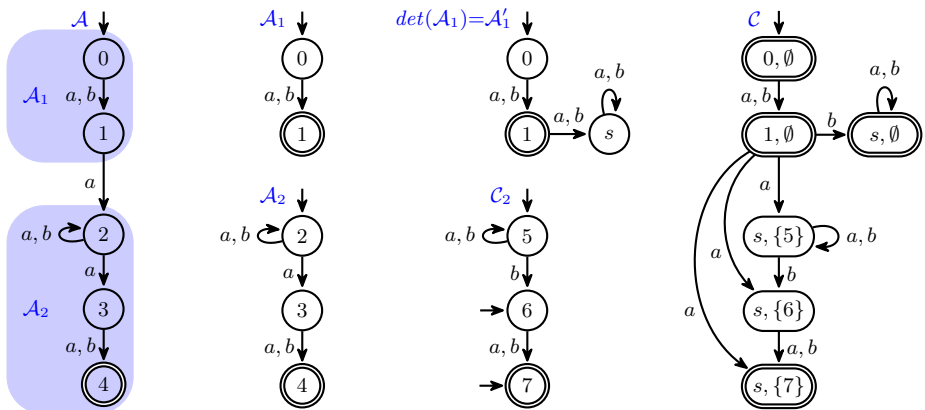

\cref{fig:port-simple-ex} shows sequential complementation of an automaton $\A$ built
from $\A_1$ and $\A_2$ connected by the transition $1\xrightarrow{a}2$ and illustrates that this
complementation can produce nondeterministic results.
% Moreover, it allows the complement $\C_2$ to be constructed and optimized in any way.
% any of the complementation methods introduced in this text can be employed.
Moreover, there exist automata for which the sequential
complementation produces complements of linear size while both forward
and reverse powerset comple\-men\-ta\-tions produce exponential
complements. Consider the language
$L_n= \{a, b\}^n.\{a\}.\{a, b\}^*.\{a\}.\{a, b\}^n$ for any
$n \in \mathbb{N}$. There exists an automaton $\B_n$ with $2n+3$
states that accepts $L_n$ (the automaton $\B_1$ is actually the
automaton $\A$ in \cref{fig:port-simple-ex}). Analogously to the figure, $\B_n$ can
be decomposed into $\B_{n1}$ and $\B_{n2}$ accepting
$L_{n1} = \{a, b\}^n$ and $L_{n2} = \{a, b\}^*.\{a\}.\{a, b\}^n$,
respectively. Complementing $\B_{n2}$ into $\C_{n2}$ via the reverse powerset and applying sequential complementation yields a complement $\C_n$ with $2n+4$ states.
In contrast, complementing $\B_n$ directly with either powerset method leads to an exponential blow-up -- both results have $2^{n+1}+n+1$~states due to the loop under $a,b$ and the nondeterminism in the middle of~$\B_n$. Moreover, $\B_n$ belongs to an NFA class for which we prove a subexponential upper bound on the sequential complement size in \ifTR\cref{sec:seq-upper-bound}\else\cite{techrep}\fi.
This upper bound is strictly better compared to forward powerset for this NFA class.

\vspace{-3mm}
\subsection{Gate Complementation} \label{sec:gate-simple}
\vspace{-1mm}
%*******************************************************************************

Recall that we consider an automaton $\A$ that can be seen as two disjoint
automata $\A_1=(Q_1,\Sigma,\delta_1,I_1,\{q_F\})$ and
$\A_2=(Q_2,\Sigma,\delta_2,\{q_I\},F_2)$ connected by a~transfer 
transition $q_F\xrightarrow{c}q_I$ for some $c\in\Sigma$, i.e.,
$\A=(Q_1\cup
Q_2,\Sigma,\delta_1\cup\delta_2\cup\{q_F\xrightarrow{c}q_I\},I_1,F_2)$.
Now we additionally assume that the symbol $c$ of the transfer
transition does not appear in any transition of the front component
$\A_1$. The transfer transition is then called a \emph{gate}
and $c$ is the \emph{gate symbol}. A scheme and an example of an NFA
$\A$ with a gate can be found in \cref{fig:gate-simple}.

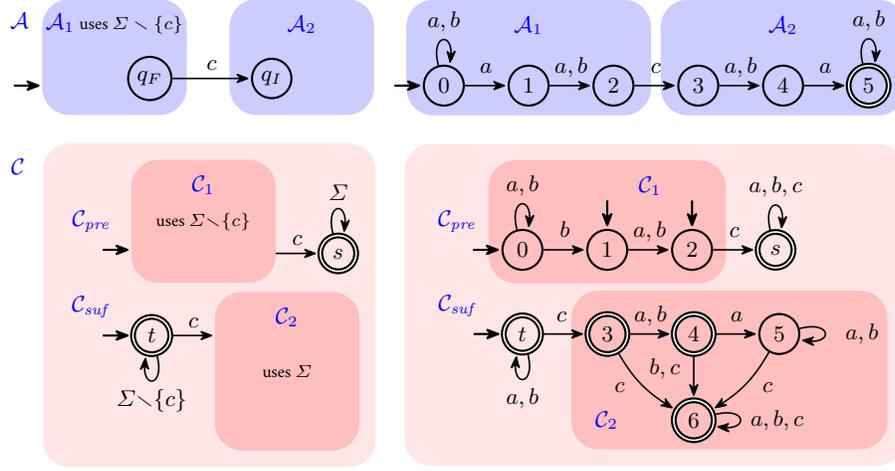
\begin{figure}[t]
    \centering

    \pgfdeclarelayer{veryback}
    \pgfsetlayers{veryback,background,main}
    \begin{tikzpicture}[scale=.95,smallautomaton,node distance=6mm,every state/.style={minimum size=5mm}]

      \tikzset{component/.style={rounded corners=10pt, fill=blue!20, draw=blue!20, draw, minimum width=2cm,semithick}}

      \tikzset{xfill/.style={rounded corners=10pt, fill=red!10, draw=red!10}}

    \begin{scope}[xshift=-1cm]
        \node[initial] (1) at (-1, 0) {};
        \node[component, anchor=west, yshift=0.4cm, minimum height=1.6cm, label={[blue,label distance =-0.6cm]90:
			$\A_1$ \textcolor{black}{\scriptsize{uses $\Sigma \smallsetminus \{c\}$}}}
		] (A1) at (1.west) {};
        \node[component, minimum height=1.6cm, label={[blue,label distance =-0.6cm]90:$\A_2$}, right = of A1] (2) {};
        \node[blue,anchor=north east,inner sep=5pt,overlay] at (A1.north west) {$\A$};

        \node[state, inner sep = 2pt, anchor=east, xshift = -0.2cm, yshift=-3mm] (p1) at (A1.east) {$q_F$};
        \node[state, inner sep = 2pt, anchor=west, xshift=0.3cm, yshift=-3mm] (p2) at (2.west) {$q_I$};
        % \node[state, double=yellow!50, anchor=east, xshift = -0.4cm, yshift=-0.3cm] (fin) at (2.east) {};
        
        \path[->, semithick]
        (p1) edge node[above] {$c$} (p2)
        ;
    \end{scope}

    \begin{scope}[xshift=0.28cm, yshift=-2.3cm, minimum width=0.88cm]
        \node[initial] (compA1) at (-0.7, 0) {};
        \node[anchor=west, component, draw=red!25, fill=red!25, minimum height=1.7cm, yshift=4mm,
			label={[blue,label distance =-0.6cm]90:$\C_1$}
		] (compA1rect) at (compA1.west) {\scriptsize{uses $\Sigma{\smallsetminus}\{c\}$}};
        \node[anchor=east,yshift=-4.5mm] (compA1end) at (compA1rect.east) {};
        \node[state, accepting, right = of compA1end] (s) {$s$};
        \node[blue, left=of compA1, yshift=4mm, xshift=5mm] (Cpre) {$\pre{\C}$};
        
        \node[state, initial, accepting, below = of compA1.west, anchor = west, yshift = -0.6cm] (t) {$t$};
        \node[blue, left=of t, yshift=4mm, xshift=5mm] (Csuf) {$\suf{\C}$};
    
        \node[right = of t] (compA2) {};
        \node[component, anchor=west,  draw=red!25, minimum height=2.2cm, yshift=-5mm, fill=red!25, inner sep=0pt,
			label={[blue,label distance =-0.6cm]90:$\C_2$}, label={[align=center]center:\scriptsize{uses $\Sigma$}}
		] (compA2rect) at (compA2.west) {};

        \begin{pgfonlayer}{background}
%            \node[xfill, fit={(Cpre) (Csuf) (compA1rect) ($(compA2rect.south east)+(.5mm,0)$)}, inner sep=6pt] (C) {};
            \node[xfill, fit={(Cpre) (Csuf) (compA1rect) ($(compA2rect.south east)+(0mm,0)$)}, inner sep=6pt] (C) {};
        \end{pgfonlayer}{background}
        \node[blue,anchor=north east, xshift=1mm, yshift=-1mm] at (C.north west) {$\C$};
        
        \path[->, semithick]
        (t) edge node[above] {$c$} (compA2)
        (compA1end) edge node[above] {$c$} (s)
        (s) edge [loop above] node {$\Sigma$} (s)
        (t) edge [loop below] node {$\Sigma{\smallsetminus}\{c\}$} (t)
        ;
    \end{scope}

    \begin{scope}[xshift=3.5cm]
        \node[state, initial] (0) {$0$};
        \node[state, right = of 0] (1) {$1$};
        \node[state, right = of 1] (2) {$2$};
        \node[state, right = of 2] (3) {$3$};
        \node[state, right = of 3] (4) {$4$};
        \node[state, accepting, right = of 4] (5) {$5$};

        \begin{pgfonlayer}{background}
            \node[component, draw=blue!20, fill=blue!20, minimum height=1.6cm, minimum width=3.4cm, yshift=0.1cm,
				label={[blue, label distance=-0.6cm]90:$\A_1$}
			] (A1) at (1.north) {};
    
            \node[component, draw=blue!20, fill=blue!20, minimum height=1.6cm, minimum width=3.4cm, yshift=0.1cm,
				label={[blue, label distance=-0.6cm]90:$\A_2$}
			] (A2) at (4.north) {};
        \end{pgfonlayer}{background}          
%        \node[blue,anchor=north east,inner sep=5pt,overlay] at ($(A2.north west)-(3.7cm,0)$) {$\A$};
        
        \path[->, semithick]
        (0) edge [loop above] node[above] {$a, b$} (0)
        (0) edge node[above] {$a\vphantom{,}$} (1)
        (1) edge node[above] {$a, b$} (2)
        (2) edge node[above] {$c\vphantom{,}$} (3)
        (3) edge node[above] {$a, b$} (4)
        (4) edge node[above] {$a\vphantom{,}$} (5)
        (5) edge [loop above] node {$a, b$} (5)
        ;
    \end{scope}

    \begin{scope}[xshift=4.6cm, yshift=-2.3cm, minimum width=1.0cm]
        \node[state, initial] (0) {$0$};
        \node[state, right = of 0, initial, initial where = above] (1) {$1$};
        \node[state, right = of 1, initial, initial where = above] (2) {$2$};
        \node[state, right = of 2, accepting] (s) {$s$};
        \node[blue, left=of 0, yshift=4mm, xshift=5mm] (Cpre) {$\pre{\C}$};

        \node[state, below = of 0, initial, accepting, yshift = 0cm] (t) {$t$};
        \node[state, right = of t, accepting] (3) {$3$};
        \node[state, right = of 3, accepting] (4) {$4$};
        \node[state, right = of 4] (5) {$5$};
        \node[state, below = of 4, accepting] (6) {$6$};
        \node[blue, left=of t, yshift=4mm, xshift=5mm] (Csuf) {$\suf{\C}$};

        \begin{pgfonlayer}{background}
            \node[component, red!25, fill=red!25, minimum height=1.7cm, minimum width=3.3cm, yshift=1mm,
				label={[blue,label distance =-0.6cm, xshift=0.6cm]90:$\C_1$}
			] (A1) at (1.north) {};
            \node[component, red!25, fill=red!25, minimum height=2.2cm, minimum width=4.4cm, xshift=5mm, yshift=-0.2cm,
				label={[blue, xshift=-0.1cm, yshift=0.65cm]225:$\C_2$}
			] (A2) at (4.south) {};
        \end{pgfonlayer}{background}
        \begin{pgfonlayer}{veryback}
            \node[xfill, fit={(Cpre) (Csuf) (A1) ($(A2.south east)+(0mm,0)$)}, inner sep=6pt] (C) {};
        \end{pgfonlayer}{veryback}
%        \node[blue,anchor=north east,yshift=-1mm] at (C.north west) {$\C$};
     
        \path[->, semithick]
        (0) edge [loop above] node {$a, b$} (0)
        (0) edge node[above] {$b\vphantom{,}$} (1)
        (1) edge node[above] {$a, b$} (2)
        (3) edge node[above] {$a, b$} (4)
        (4) edge node[above] {$a\vphantom{,}$} (5)
        (5) edge [bend left=15] node [below, pos=0.3, xshift=5pt] {$c\vphantom{,}$} (6)
        (4) edge node[left, pos=0.3, xshift=3pt] {$b, c$} (6)
        (3) edge [bend right=15] node [below, pos=0.3, xshift=-5pt] {$c\vphantom{,}$} (6)
        (5) edge [loop right] node {$a, b$} (5)
        (6) edge [loop right] node {$a, b, c$} (6)
        (t) edge [loop below] node {$a, b$} (t)
        (s) edge [loop above] node {$a, b, c$} (s)
        (2) edge node[above] {$c\vphantom{,}$} (s)
        (t) edge node[above, xshift=-1pt] {$c\vphantom{,}$} (3)
        ;
    \end{scope}
    \end{tikzpicture}

    \caption{A generic scheme of an NFA $\A$ with a gate (top left), a scheme of
      its complement~$\C$ (bottom left), an instance of $\A$ for language
      $\{a, b\}^*.\{a\}.\{a, b\}.\{c\}.\{a, b\}.\{a\}.\{a, b\}^*$ (top
      right), and its gate complement $\C$ (bottom right).}
    \label{fig:gate-simple}
    \vspace*{-3mm}
\end{figure}

Gate complementation utilizes the specific properties of $\A$.  The
automaton $\A$ accepts the language $L_1.\{c\}.L_2$ where
$L_1=\langof{\A_1}$ and $L_2=\langof{\A_2}$.  A word
$w \in\Sigma^*$ belongs to
$\compl{L_1.\{c\}.L_2}$ in the following cases:
% \begin{enumerate}
%     \item $w$ does not contain any $c$. \label{item:gate-nogsym}
%     \item $w = ucv$, where $u \in (\Sigma\smallsetminus\{c\})^*$, and either 
%     \begin{enumerate}
%         \item $u \notin L_1$, or \label{item:gate-pre}
%         \item $v \notin L_2$. \label{item:gate-suf}
%     \end{enumerate}
% \end{enumerate}
\begin{enumerate}
    \item $w$ does not contain any $c$, \label{item:gate-nogsym}
    \item $w = ucv$ where $u \in (\Sigma\smallsetminus\{c\})^*$ and $u \notin L_1$, or \label{item:gate-pre}
    \item $w = ucv$ where $u \in (\Sigma\smallsetminus\{c\})^*$ and $v \notin L_2$. \label{item:gate-suf}
\end{enumerate}

Let
$\C_1 = (\elc{Q}_1, \Sigma\smallsetminus\{c\}, \elc{\delta}_1,
\elc{I}_1, \elc{F}_1)$ be an arbitrary complement of $\A_1$ with
respect to the alphabet $\Sigma\smallsetminus\{c\}$ and
$\C_2 = (\elc{Q}_2, \Sigma, \elc{\delta}_2, \elc{I}_2, \elc{F}_2)$ be
an arbitrary complement of~$\A_2$ with respect to the alphabet
$\Sigma$. Then, the complement~$\C$ of~$\A$ consists of two parts:
$\pre{\C}$~accepting the words according to Case~\ref{item:gate-pre} and
$\suf{\C}$~accepting the words according to Cases~\ref{item:gate-nogsym}
and~\ref{item:gate-suf}. Schemes of $\C$, $\pre{\C}$, and $\suf{\C}$ can be seen in
\cref{fig:gate-simple}.

Formally, we set $\C = \pre{\C} \Cup \suf{\C}$ where
$\pre{C} = (\elc{Q}_1 \cup \{s\}, \Sigma, \pre{\delta}, \elc{I}_1,
\{s\})$ and
$\suf{C} = (\elc{Q}_2 \cup \{t\}, \Sigma, \suf{\delta}, \{t\}, \{t\}
\cup \elc{F}_2)$ such that $s,t$ are fresh states,
\begin{align*}
    \pre{\delta} &= \elc{\delta}_1 \cup \{p \xrightarrow{c} s \mid p \in \elc{F}_1\} \cup \{s \xrightarrow{a} s \mid a \in \Sigma\} \text{, and } \\
    \suf{\delta} &= \elc{\delta}_2 \cup \{t \xrightarrow{c} p \mid p \in \elc{I}_2\} \cup \{t \xrightarrow{a} t \mid a \in \Sigma\smallsetminus\{c\}\}.
\end{align*}

%The correctness of the construction is proven in \ifTR\cref{app:proofs}\else\cite{techrep}\fi.

\begin{theorem}\label{thm:gate-basic}
  The NFA $\C$ accepts $\compl{\langof{\A}}$.
\end{theorem}
\begin{proof}[Proof (sketch)]
  Recall that $\langof{A}=L_1.\{c\}.L_2$. First, consider
  $w\in\langof{\A}$. Then $w = ucv$, where $u \in L_1$ and
  $v \in L_2$. Because $u \in L_1$, no run of $\C_1$ over $u$ reaches
  $\elc{F}_1$. Hence, no run of~$\pre{\C}$ over a word starting with
  $uc$ reaches $s$ and thus $w\not\in\langof{\pre{\C}}$. The run of
  $\suf{\C}$ over $uc$ reaches $\elc{I}_2$ and it cannot be prolonged
  into an accepting run over $w$ as $\C_2$ does not accept~$v$.
  Altogether, we get $w\not\in\langof{\pre{\C}}\cup\langof{\suf{\C}}=\langof{\C}$.

  Now assume that $w\not\in\langof{\A}$. If $w$ does not contain any
  $c$, it is accepted by $\suf{\C}$. Let $w=ucv$, where
  $u\in(\Sigma\smallsetminus\{c\})^*$. As $w\not\in\langof{\A}$, we
  know that $u\not\in L_1$ or $v\not\in L_2$. In the former case,
  $u\in\langof{\C_1}$ and thus $w$ is accepted by $\pre{\C}$. In the
  latter case, $v\in\langof{\C_2}$ and thus $w$ is accepted by $\suf{\C}$.
  To sum up, $w\in\langof{\pre{\C}}\cup\langof{\suf{\C}}=\langof{\C}$.
\end{proof}
%\vspace{-2mm}
% \begin{theorem}\label{thm:gate-basic}
%   The NFA $\C$ constructed above accepts $\compl{\langof{\A}}$.
% \end{theorem}
%\vspace{-2mm}
% \begin{proof}[Proof (sketch)]
%   Recall that $\langof{A}=L_1.\{c\}.L_2$. First, consider
%   $w\in\langof{\A}$. Then $w = ucv$, where $u \in L_1$ and
%   $v \in L_2$. Because $u \in L_1$, no run of $\C_1$ over $u$ reaches
%   $\elc{F}_1$. Hence, no run of~$\pre{\C}$ over a word starting with
%   $uc$ reaches $s$ and thus $w\not\in\langof{\pre{\C}}$. The run of
%   $\suf{\C}$ over $uc$ reaches $\elc{I}_2$ and it cannot be prolonged
%   into an accepting run over $w$ as $\C_2$ does not accept~$v$.
%   Altogether, we get $w\not\in\langof{\pre{\C}}\cup\langof{\suf{\C}}=\langof{\C}$.

%   Now assume that $w\not\in\langof{\A}$. If $w$ does not contain any
%   $c$, it is accepted by $\suf{\C}$. Now assume that $w=ucv$, where
%   $u\in(\Sigma\smallsetminus\{c\})^*$. As $w\not\in\langof{\A}$, we
%   know that $u\not\in L_1$ or $v\not\in L_2$. In the former case,
%   $u\in\langof{\C_1}$ and thus $w$ is accepted by $\pre{\C}$. In the
%   latter case, $v\in\langof{\C_2}$ and thus $w$ is accepted by $\suf{\C}$.
%   To sum up, $w\in\langof{\pre{\C}}\cup\langof{\suf{\C}}=\langof{\C}$
% \end{proof}

% The advantage of gate complementatin is that an arbitrary
% complementing algorithm may be chosen for each of the components.

Note that $\C$ has only $|\C_1|+|\C_2|+2$ states. To see the advantage
of this complementation approach, consider
$L_n=L_{n1}.\{c\}.L_{n2}$, where
$L_{n1} = \{a, b\}^*.\{a\}.\{a, b\}^n$,
$L_{n2} = \{a, b\}^n .\{a\}.\{a, b\}^*$, and $n \in \mathbb{N}$.
There exists an NFA $\B_n$ with only $2n+4$ states accepting~$L_n$, which can be deconstructed into two components accepting $L_{n1}$ and $L_{n2}$, respectively
(the automaton $\B_1$ is actually the automaton $\A$ in
\cref{fig:gate-simple}). If we use reverse powerset to complement the front component, forward powerset for the rear one, and combine the outputs by gate
complementation, the result $\C_n$ has $2n+7$ states. Complementing the whole NFA $\B_n$ with either forward or
reverse powerset causes an exponential blow-up, with both results having $2^{n+1} + n + 2$ states. Sequential complementation (\cref{sec:seq-simple}) also results in
a~blow-up (for all possible divisions of~$\B_n$
into front and rear components) due to the determinization of the
front component and/or tracking possibly many instances of the
complement of the rear component.

% This idea is further generalized to a situation where the two
% components $\A_1$ and $\A_2$ are connected by multiple transitions
% under multiple different special symbols. We call such automata
% \emph{gate automata} and present two generalized complementation
% algorithms for two distinct types of gate automata in
% \cref{sec:gate-general}.

%}}}
%{{{ generalized problem

%%%%%%%%%%%%%%%%%%%%%%%%%%%%%%%%%%%%%%%%%%%%%%%%%%%%%%%%%%%%%%%%%%%%%%%%%%%%%%%%
\vspace{-3.0mm}
%\section{Generalized Complementation Problem} \label{sec:generalized-compl}
\section{Generalized Complementation Problem}\label{sec:generalized-compl} 
\vspace{-2.0mm}
%%%%%%%%%%%%%%%%%%%%%%%%%%%%%%%%%%%%%%%%%%%%%%%%%%%%%%%%%%%%%%%%%%%%%%%%%%%%%%%%

This section briefly generalizes the ideas from
\cref{sec:port-simple}. The generalized complementation constructions
are still applicable to automata consisting of two components, but
there can be an arbitrary number of transfer transitions under various
symbols leading from the front to the rear component. The generalized
constructions again use complements of the components. These
complements can be constructed either by forward or reverse powerset
complementation, or by a recursive application of (generalized)
sequential or gate complementation.

Due to the potentially recursive application of our complementation
constructions and due to the fact that components can be connected by
multiple transfer transitions, we need to complement components with
many incoming and outgoing transfer transitions. Therefore, we work
with automata that generalize initial and final states to multiple
sets of entry and exit states. These sets are called \emph{entry} and
\emph{exit port sets}. We talk about \emph{port automata} and
generalize the complementation problem to these automata as follows.

% Intuitively, we decompose a given NFA into
% potentially many components, where each component is an SCC or a union
% of some SCCs with the corresponding transfer transitions. In contrast
% to \cref{sec:port-simple}, we assume that these components are
% connected by arbitrarily many transfer transitions.

% Roughly speaking, we repeatedly
% select two of the components, merge them into one and build the
% complement of this merged component by a generalized sequential or
% gate complementation. We repeat this process until we obtain a
% complement of the original automaton.

% As each component can have many incoming and outgoing transfer
% transitions, we work with automata that generalize initial and final
% states to multiple sets of entry and exit states. These sets are called
% \emph{entry} and \emph{exit port sets}. We talk about \emph{port
%   automata} and generalize the complementation problem to these
% automata in the following way.

A~\emph{port NFA} or simply a \emph{port automaton} is a tuple
$\A = (Q, \Sigma, \delta, \inportsets, \outportsets)$, where $Q$,
$\Sigma$, and~$\delta$ are as in an NFA, and
$\inportsets = (I_0,\ldots,I_k)$ and $\outportsets = (F_0,\ldots,F_{\ell})$
are sequences of subsets of $Q$ called \emph{entry port sets},
resp.~\emph{exit port sets}.  A~\emph{slice} of $\A$ is an NFA with
one entry and one exit port set chosen as the initial, resp.\ final
states, i.e., the NFA
$\autij{\A}{i}{j} = (Q, \Sigma, \delta, I_i, F_j)$ for $0 \leq i \leq k$ and
$0 \leq j \leq \ell$.  $\A$~is \emph{deterministic (port DFA)} if all its
slices are deterministic (in particular, $|I_i| = 1$ must hold for every $0 \leq i \leq k$).
A port DFA is \emph{complete} if all its slices are complete.
A~\emph{complement} of $\A$ is a port NFA representing
complements of all slices. More precisely, a~complement of $\A$ is a port NFA
$\C = (\elc{Q}, \Sigma, \elc{\delta}, \elc{\inportsets},
\elc{\outportsets})$ with
$\elc{\inportsets} = (\elc{I}_0,\ldots,\elc{I}_k)$ and
$\elc{\outportsets} = (\elc{F}_1,\ldots,\elc{F}_{\ell})$ such that 
$\langof{ \autij{\C}{i}{j} } = \compl{ \langof{ \autij{\A}{i}{j} }}$
for each $0 \leq i \leq k$ and $0 \leq j \leq \ell$. We call~$\elc{I}_i$
an \emph{entry port complement} of $I_i$ and $\elc{F}_j$
an \emph{exit port complement} of $F_j$, together shortened to
\emph{port complements}. 

In the rest of this section, we first generalize the powerset
construction to port automata to get a determinization procedure
needed in the sequential complementation. With that, we generalize the
forward and reverse powerset complementation to port automata. Powerset
complementations are applied to components that cannot be recursively
complemented by sequential or gate complementation (for example,
because they cannot be further decomposed into a front and a rear
component). Finally, we outline general versions of sequential and
gate complementations.

\vspace{-0.0mm}
\subsection{Powerset Construction and Complementation for Port Automata} \label{sec:extended-powerset}
\vspace{-0.0mm}
%*******************************************************************************

%The classical powerset construction and complementation may be generalized to port automata as follows.
%Given a port NFA $\A = (Q, \Sigma, \delta, \inportsets = (I_i)_{0\leq i <k}, \outportsets= (F_j)_{0\leq j<l})$, 
%the \emph{port powerset construction} produces a port DFA $\det{\A} = (Q', \Sigma, \delta', \inportsets' = (I_i')_{0\leq i<k}, \outportsets' = (F_j')_{0\leq j<l})$, where $Q'$ and $\delta'$ are defined as in \cref{sec:forward-powerset}, $I'_i = \{I_i\}$ and $F'_j = \{P \in Q' \mid P \cap F_j \neq \emptyset\}$, $0\leq i \leq k,0\leq j \leq \ell$.
%The original $\A$ and $\det{\A}$ are equivalent, meaning that $\langof{\autij{\A}{i}{j}} = \langof{\autij{\det{\A}}{i}{j}}$, $0\leq i<k,0\leq j<l$. 
%%
%The \emph{complement} of the port DFA $\D$ is a port DFA $\co{\D} = (Q', \Sigma, \delta', \inportsets', \outportsets'')$, where $\outportsets''=(Q\smallsetminus F_j)_{0\leq j<l}$.
%
%Forward powerset complement of a port NFA is, as for non-port NFAs, defined as $\co{\det{\A}}$. Reverse powerset complement is again constructed as $\rev{\co{\det{\rev{\A}}}}$.

%%%%%%%%%%%%%%%%%

The powerset construction and complementation generalize to
port automata as follows.  Given a port NFA
$\A = (Q, \Sigma, \delta, \inportsets, \outportsets)$ with
$\inportsets = (I_i)_{0\leq i \leq k}$ and $\outportsets= (F_j)_{0\leq j \leq
  \ell}$, the \emph{port powerset construction} produces a port DFA
$\det{\A} = (Q', \Sigma, \delta', \inportsets', \outportsets')$, where
$Q'$ and $\delta'$ are defined as in the standard powerset construction
(see \cref{sec:forward-powerset}),
$\inportsets' \!=\! (\{I_i\})_{0\leq {i}\leq k}$, and $\outportsets'\! = \!(\{P \!\in
\! Q' \mid P \!\cap\! F_j \!\neq\! \emptyset\})_{0\leq j \leq \ell}$.
%$I'_i = \{I_i\}$ and $F'_j = \{P \in Q' \mid P \cap F_j \neq \emptyset\}$, $0\leq i \leq k,0\leq j \leq \ell$.
%
The original $\A$ and $\det{\A}$ are equivalent, i.e., $\langof{\autij{\A}{i}{j}} = \langof{\autij{\det{\A}}{i}{j}}$ for all $0\leq i \leq k$ \mbox{and $0\leq j \leq \ell$.}
Moreover, $\det{\A}$ is \emph{complete}.
The \emph{complement} of any complete port DFA
$\D = (Q', \Sigma, \delta', \inportsets', \outportsets')$ with $\outportsets'= (F'_j)_{0\leq j \leq
  \ell}$ is the port
DFA $\co{\D} = (Q', \Sigma, \delta', \inportsets', \outportsets'')$
where $\outportsets''=(Q'\smallsetminus F'_j)_{0\leq j \leq \ell}$.
% \js{zduraznit, ze tohle funguje na libovolny *complete* port DFA, nejen vysledek $\det()$, definovat complete?}
% 
The forward powerset complement of a port NFA is, as for non-port NFAs, defined as $\co{\det{\A}}$.
The \emph{reverse} of a port NFA
$\A$ is the port NFA
$\rev{\A}=(Q, \Sigma, \rev{\delta}, \outportsets,\inportsets)$.
% , where
% $\rev{\delta}$ is defined in
% \cref{sec:forward-powerset}.
The reverse powerset complement is then
constructed as $\rev{\co{\det{\rev{\A}}}}$.

%%%%%%%%%%%%%%%%

\vspace{-0mm}
\subsection{Generalized Sequential Complementation} \label{sec:port-general-outline}
\vspace{-0mm}

We now outline the generalization of sequential complementation to
port NFAs, highlighting only the differences from
\cref{sec:seq-simple} (see \ifTR\cref{sec:port-general}\else\cite{techrep}\fi{}
for details).
Let $\A$ be a component constructed by merging $\A_1$ and
$\A_2$. Hence, $\A$~is a port NFA consisting of the front component $\A_1$ and the rear component $\A_2$, both port NFAs, 
connected by a~set of transfer transitions $\gates$ leading from $\A_1$ to  $\A_2$. 
Unlike \cref{sec:seq-simple}, we have a~set of transfer transitions $\gates$ instead of a~single one, and both components have multiple entry/exit port sets. 
%We are to construct a complement $\C$ of $\A$.
We first construct $\det{\A_1}$ and a~complement port NFA~$\C_2$ of~$\A_2$.
Analogously to \cref{sec:seq-simple}, the constructed complement~$\C$ of~$\A$ contains states $(q,R)$, where $q$ tracks the only run of % a slice of
$\det{\A_1}$ and $R$ tracks runs of $\C_2$ over suffixes of the input word. 

The generalized construction closely follows the basic one. It differs in the following: %\js{nemelo by se v dalsim misto $\A_1$ vsude psat $\det{\A_1}$?}
\begin{enumerate}
\item Given a state $q$ of $\det{\A_1}$ and a symbol $a\in\Sigma$, let
  $P$ be the set of states $p$ of $\A_2$ such that $\A$ contains a
  transfer transition $q''\xrightarrow{a} p$ from some $q''\in q$.
  Whenever $\C$ reaches a state $(q, R)$ with $a$ next on input, 
  it has to check that the rest of the word is not accepted by $\A_2$ starting from any
  $p\in P$ and thus it spawns a new instance of $\C_2$ for each
  $p\in P$. Formally, $\C$ has all transitions
  $(q, R)\xrightarrow{a}(q',R')$ where (1)~$q \xrightarrow{a}q'$ is
  the transition from $q$ under $a$ in $\det{\A_1}$, (2)~for each
  $r\in R$, $R'$ contains some $s$ such that $r \xrightarrow{a}s$ is a
  transition in $\C_2$, and, additionally, (3)~for every $p \in P$,
  $R'$ contains some state $s$ from the port complement of the newly
  added entry port set $\{p\}$. % (cf.\ the construction of~$\C_2$).
% be the set of those inner entry port sets. % and let $X'$ be the set of their port complements. 
% $\C$ must now check that the rest of the word is accepted in $\C_2$ from the port complements of all the port sets in $X$.

% (while in \cref{sec:port-simple}, we were adding only the states of $s_0\in\elc{I}_2$).
%To do that, $\C$ will continue with a nondeterministic choice of $a$-transitions from $(q,R)$ leading to all states of the form $(q',R')$,
%where $q'$ is the $a$-successor of $q$, and $R'$ contains, besides the $a$-successors of $R$, also an element of every port set $X'$ (in \cref{sec:port-simple}, we were adding only the states of $s_0\in\elc{I}_2$).
\item
% $\C$ has entry and exit ports defined by a natural generalization of the definition of initial and final states in \cref{sec:port-general}. 
  Since $\A_1$ may have exit ports with transitions leaving $\A$ and
  $\A_2$ may have entry ports with incoming edges from outside, $\C$
  must reject words accepted entirely within either $\A_1$ or
  $\A_2$. Therefore, if a slice $\autij{\A}{i}{j}$ has entry ports in
  $\A_2$, we activate one instance of $\autij{\C_2}{i}{j}$ at the
  start: $\autij{\C}{i}{j}$ has entry ports of the form $(q,\{r\})$,
  where $q$ is the only entry port in $\autij{\det{\A_1}}{i}{j}$ and $r$ is an entry port in $\autij{\C_2}{i}{j}$. %\js{a co je $q$?} 
  At the same
  time, a state $(q, R)$ is an exit port of $\autij{\C}{i}{j}$ only if
  $q$ is not an exit port of $\autij{\det{\A_1}}{i}{j}$, ensuring the input
  word was not accepted in $\autij{\A_1}{i}{j}$.
% and if all states in $R$ are outer exit ports of $\autij{\C_2}{i}{j}$. 
\end{enumerate}

%*******************************************************************************
\vspace{-2mm}
\subsection{Generalized Gate Complementation}\label{sec:gate-general-outline}
\vspace{-1mm}
%*******************************************************************************

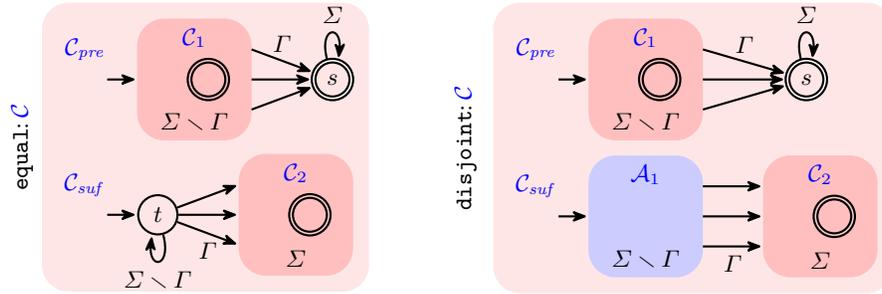
\begin{figure}[t]
    \centering
    \begin{tikzpicture}[smallautomaton,node distance=8mm,every state/.style={minimum size=5mm}]

    \tikzset{component/.style={rounded corners=10pt, draw=blue!20, fill=blue!20, draw, minimum width=1.5cm, minimum height=1.6cm, semithick}}

    \tikzset{xfill/.style={rounded corners=10pt, fill=red!10}}

    \begin{scope}
        % Cpre
        \node[initial, component, fill=red!25, draw=red!25, label={[blue, label distance =-0.5cm]90:$\C_1$}] (compA1) {};
        \node[state, accepting, double=red!25, anchor=east, xshift = -3mm] (fin1) at (compA1.east) {};
        \node[blue, left=of compA1, yshift=4mm, xshift=5mm] (Cpre) {$\pre{\C}$};
        
        \node[state, accepting, right = of compA1] (s) {$s$};
        
        \node[anchor=south] (A1sym) at (compA1.south) {$\Sigma \smallsetminus \gsym$};
        
        \node[coordinate, yshift=4mm] (7) at (compA1.east) {};
        \node[coordinate] (8) at (compA1.east) {};
        \node[coordinate, yshift=-4mm] (9) at (compA1.east) {};
        
        \path[->]
        (7) edge node[above] {$\gsym$} (s)
        (8) edge node {} (s)
        (9) edge node {} (s)
        (s) edge [loop above] node {$\Sigma$} (s)
        ;

        % Csuf
        \node[state, initial, below = of compA1.west, anchor = west, yshift = -1cm] (t) {$t$};
        \node[component, fill=red!25, draw=red!25, right = of t, label={[blue, label distance =-0.5cm]90:$\C_2$}] (compA2) {};
        \node[state, accepting, double=red!20, anchor=east, xshift = -3mm] (fin2) at (compA2.east) {};

        \node[blue, left=of t, yshift=4mm, xshift=5mm] (Cpre) {$\suf{\C}$};
        \node[anchor=south] (A2sym) at (compA2.south) {$\Sigma$};
        
        \node[coordinate, yshift=4mm] (4) at (compA2.west) {};
        \node[coordinate] (5) at (compA2.west) {};
        \node[coordinate, yshift=-4mm] (6) at (compA2.west) {};
        
        \path[->]
        (t) edge node {} (4)
        (t) edge node {} (5)
        (t) edge node [below] {$\gsym$} (6)
        (t) edge [loop below] node {$\Sigma \smallsetminus \gsym$} (s)
        % (r) edge node {} (rout)
        ;

        \begin{pgfonlayer}{background}
            \node[xfill, fit = (compA1) (s) (A1sym) (compA2) (t) (A2sym), fit margins={left=18pt,right=3pt,bottom=3pt,top=3pt}] (C) {};
        \end{pgfonlayer}{background}
        \node[anchor=south, rotate=90] at (C.west) {\methodEqual: {\color{blue}$\C$}};

    \end{scope}

    \begin{scope}[xshift=6cm]
        % Cpre
        \node[initial, component, fill=red!25, draw=red!25, label={[blue, label distance =-0.5cm]90:$\C_1$}] (compA1) {};
        \node[state, accepting, double=red!25, anchor=east, xshift = -3mm] (fin1) at (compA1.east) {};
        \node[blue, left=of compA1, yshift=4mm, xshift=5mm] (Cpre) {$\pre{\C}$};
        
        \node[state, accepting, right = of compA1, xshift=3mm] (s) {$s$};
        
        \node[anchor=south] (A1sym) at (compA1.south) {$\Sigma \smallsetminus \gsym$};
        
        \node[coordinate, yshift=4mm] (7) at (compA1.east) {};
        \node[coordinate] (8) at (compA1.east) {};
        \node[coordinate, yshift=-4mm] (9) at (compA1.east) {};
        
        \path[->]
        (7) edge node[above] {$\gsym$} (s)
        (8) edge node {} (s)
        (9) edge node {} (s)
        (s) edge [loop above] node {$\Sigma$} (s)
        ;

        % Csuf
        \node[initial, draw=blue!20, component, below = of compA1, yshift=6mm, label={[blue, label distance =-0.5cm]90:$\A_1$}] (A1) {};
        \node[component, fill=red!25, draw=red!25, right = of A1, label={[blue, label distance =-0.5cm]90:$\C_2$}] (compA2) {};
        \node[state, accepting, double=red!20, anchor=east, xshift = -3mm] (fin2) at (compA2.east) {};
        \node[blue, left=of A1, yshift=4mm, xshift=5mm] (Cpre) {$\suf{\C}$};
        
        \node[anchor=south] (A2sym) at (compA2.south) {$\Sigma$};
        \node[anchor=south] (A2sym) at (A1.south) {$\Sigma \smallsetminus \gsym$};
    
        \node[coordinate, yshift=4mm] (1) at (A1.east) {};
        \node[coordinate] (2) at (A1.east) {};
        \node[coordinate, yshift=-4mm] (3) at (A1.east) {};
        \node[coordinate, yshift=4mm] (4) at (compA2.west) {};
        \node[coordinate] (5) at (compA2.west) {};
        \node[coordinate, yshift=-4mm] (6) at (compA2.west) {};
         
        \path[->]
        (1) edge node {} (4)
        (2) edge node {} (5)
        (3) edge node [below] {$\gsym$} (6)
        ;

        \begin{pgfonlayer}{background}
            \node[xfill, fit = (compA1) (s) (A1sym) (compA2) (A1), fit margins={left=18pt,right=3pt,bottom=3pt,top=3pt}] (C) {};
        \end{pgfonlayer}{background}
        \node[anchor=south,xshift=-1mm,rotate=90] at (C.west) {\methodDisjoint: {\color{blue}$\C$}};
    \end{scope}

    \end{tikzpicture}
    \caption{Gate complement using the \methodEqual (left) and \methodDisjoint method (right).}
    \label{fig:gate-comp-general}
\vspace{-5mm}
\end{figure}

Here we outline the core ideas of generalized gate complementation as
an extension of \cref{sec:gate-simple}, with full details in
\ifTR\cref{sec:gate-general}\else\cite{techrep}\fi. Recall that in \cref{sec:gate-simple}, the
NFA $\A$ was assumed to consist of two components connected by a
single gate transition labeled with a gate symbol not occurring in the
front component.  We now generalize this by allowing arbitrary
transitions leading from the front to the rear component and labeled with gate symbols $\Gamma\subseteq\Sigma$ such that $\A_1$ has no transitions under $\Gamma$.
% , with a set~$\gsym$ of multiple different gate symbols.
We call a gate transition labeled with $c$ a $c$-gate. We also suppose that
the input NFA may be a port NFA.
We present two variants of the algorithm, \methodEqual and
\methodDisjoint (illustrated in \cref{fig:gate-comp-general}), each adding further constraints on the input NFA.

The simpler variant \methodEqual more closely resembles
\cref{sec:gate-simple}. 
If we write the words read by $\A$ as $ucv$ with $u \in (\Sigma \smallsetminus \gsym)^*$, $c \in \gsym$, and $v \in \Sigma^*$, this variant
assumes that for every gate symbol
$c \in \gsym$, every two entry ports of $\A_2$ with an 
incoming $c$-gate can be reached by reading exactly the same 
prefixes of the form $uc$.
% \js{nasledujici veta je nejasna} 
% This means that any word read in $\A_1$ followed by $c$ can be paired
% with any word read in $\A_2$ after taking a gate $c$-transition.
%
To recognize all prefixes that cannot be read in
$\A_1$ and followed by $c$, we collect the ports of $\A_1$ with outgoing $c$-gates into a new exit port set. We do the same for
the suffixes that are not accepted in $\A_2$ from states with an incoming $c$-gate, and give~$\A_2$ a~new entry port set
consisting of those states. These new port
sets (one for each gate symbol) are preserved through complementation,
and in $\pre\C$ and $\suf\C$, we can connect their complement port
sets in~$\C_1$ and~$\C_2$ to the states $s$ and $t$ using an
appropriate gate symbol, respectively.
%This way, all words starting with an invalid prefix followed by $c$ are accepted in $\pre\C$, and those ending with an invalid suffix preceded by $c$ are recognized in $\suf\C$. 
This way, $\pre\C$ accepts words where the first-appearing gate symbol $c$ follows an invalid prefix
and $\suf\C$ accepts words where $c$ is followed by an invalid suffix.

% The language of the original automaton can hence be partitioned according to the (first) gate symbol $c$. 
% In the partition corresponding to $c$, $c$ is preceded by a language of prefixes $L_1^c$ accepted in $\A_1$ and followed by the language of suffixes $L_2^c$ accepted in $\A_2$.
% Hence, in the complemented parts $\pre\C$ and $\suf\C$, 
% we must connect prefixes $\compl{L^1_c}$ and suffixes $\compl{L^2_c}$ to the states $s$ and $t$, respectively, 
% via a transition labeled by the same gate symbol $c$.   
% %
% Therefore, before complementing $\A_1$ and $\A_2$, 
% we group origins of the gate transitions in $\A_1$ into new port sets according to the gate transition symbol 
% and group targets of gate transitions in $\A_2$ into new port sets according to the gate transition symbol.
% These new port sets are preserved through complementation, 
% and in $\pre\C$ and $\suf\C$, we can connect their complement port sets in $\C_1$ and $\C_2$ to \mbox{the states $s$ and $t$ using appropriate~ gate symbol.} 

% and \methodDisjoint. Both methods are applicable to a restricted class of gate NFAs. The respective classes for both methods are incomparable. At the end of this section, we also present several modifications of the two main methods, which further extend the classes of automata we are able to complement.

%The \methodDisjoint variant then relaxes the requirement on the equality of $\A_2$ entry-port languages. 
The \methodDisjoint variant relaxes the requirement that all $\A_2$'s entry ports with an incoming $c$-gate must be reached by the same words of the form $uc$. 
Instead, $\suf{\C}$ tracks used gates, assuming the complemented port
NFA is partitioned as follows: for any $c \in \gsym$ and any two gates $p \xrightarrow{c} r$, $p' \xrightarrow{c} r'$, if the languages accepted by $p$ and $p'$ in any slice
$\autij{\A_1}{i}{j}$ are not disjoint, then the languages
accepted from $r$ and $r'$ in the slice $\autij{\A_2}{i}{j}$ must be equal. 
% \js{nebo disjoint?} \adela{equal je spravne (pokud nejsou jazyky pred gatami disjoint, musi byt jazyky za gatami equal)}. 
Unlike the \methodEqual method, $\suf{\C}$ now depends on both $\A_1$ and $\C_2$ to \mbox{track the used gate.}

%}}}
%{{{ evaluation

%%%%%%%%%%%%%%%%%%%%%%%%%%%%%%%%%%%%%%%%%%%%%%%%%%%%%%%%%%%%%%%%%%%%%%%%%%%%%%%%
\vspace{-3.0mm}
\section{Implementation and Evaluation}\label{sec:experiments}
\vspace{-2.0mm}
%%%%%%%%%%%%%%%%%%%%%%%%%%%%%%%%%%%%%%%%%%%%%%%%%%%%%%%%%%%%%%%%%%%%%%%%%%%%%%%%

We have implemented the described algorithms in a tool called
\aligater\footnote{\blinded{\url{https://gitlab.fi.muni.cz/xstepkov/aligater}}}
written in Python and using the C++ library
\mata~\cite{Mata} and the Python library \texttt{Automata}~\cite{Automatalib} as
backends. \aligater~also integrates the \rabit tool~\cite{MayrC13} for NFA reduction.

\aligater~calls \mata~for forward (\fwdpws) and reverse (\revpws) powerset complementation, and Hopcroft's
minimization~\cite{Hopcroft71,ValmariL08} (denoted by the suffix \plusmin).
Sequential (\sequential) and gate (\gate) complementations are implemented in \aligater in their
generalized versions, described in \cref{sec:port-general-outline,sec:gate-general-outline}.
% \ol{do we not want to refer to \cref{sec:port-general-outline,sec:gate-general-outline}?}
%The following sections discuss important implementation details for both of these methods.
Selected implementation details for both of these methods are discussed below.
Other details (e.g., the use of reductions) and
an extended description of \aligater settings used in the evaluation are available
in \ifTR\cref{sec:impl-eval-settings}\else\cite{techrep}\fi.

%*******************************************************************************
\vspace{-0.0mm}
\subsection{Implementation of Sequential Complementation}\label{sec:impl-seq}
\vspace{-0.0mm}
%*******************************************************************************

% \js{Zacatek toho textu by se hodil spis do sekce 5.}
% The generalized sequential complementation, fully defined in \cref{sec:port-general}, can complement any NFA $\A$ divided into two components $\A_1$ and $\A_2$. It only requires that transitions between both components are just in the direction from $\A_1$ to $\A_2$. The algorithm determinizes $\A_1$, complements $\A_2$, and composes both together into a complement of $\A$. The complement of $\A_2$ is obtained by recursively partitioning and complementing it using the same method.
% \js{Az po sem.}

The implementation of sequential complementation first divides the NFA $\A$ into components
$\A_1, \ldots, \A_n$.  Then $\A_n$ is complemented using forward or
reverse powerset complementation, and $\A_1, \ldots, \A_{n-1}$ are
determinized.
% Sequential composition is then applied to $\det{\A_{n-1}}$ and the complement
% of~$\A_n$ to form a new complemented bottom component;
The automaton~$\det{\A_{n-1}}$ is then composed with the complement
of~$\A_n$ to form a new complemented bottom component;
the same
process is repeated with every preceding component up to~$\A_1$.

We implemented three approaches to divide $\A$ into components:
\begin{enumerate}
\item \emph{Deterministic components}. Because the components
  $\A_1, \ldots, \A_{n-1}$ are determinized during sequential
  complementation, this partitioning approach tries to avoid
  determinization and use transfer transitions to cover some
  nondeterminism in $\A$. We first decompose $\A$ into SCCs and order
  them by topological ordering. If the first SCC is nondeterministic,
  we set $\A_1$ to be this SCC. If it is deterministic, we set $\A_1$
  to be the first SCC and we repeatedly add previously unused SCCs in the order of the topological ordering (with
  the corresponding transfer transitions) until a maximal
  deterministic $\A_1$ is obtained. Note that $\A_1$ may not be connected. The rest of the automaton $\A$ is then recursively divided in the same
  manner.
\item \emph{Deterministic components + reverse-deterministic bottom
    component}. Sequential complementation requires~$\A_n$ to be
  complemented, and if it is reverse-deterministic, we can complement it
  easily by reverse powerset complementation.
  % the strategy here is to, firstly, avoid the complement of $\A_n$
  % being too
  % large, and secondly, introduce more nondeterminism into the
  % complement of
  % $\A$ by using reverse powerset construction on $\A_n$.
  We compute $\A_n$ as a maximal
  bottom reverse-deterministic part of $\A$, analogously to the
  computation of $\A_1$ in the previous approach. The rest of $\A$ is
  divided in the same way as above.
  % divided in the same way as in the previous approach.
\item \emph{Min-cut}. $\A$ is divided into two components, where the
  partition with the fewest transfer transitions from $\A_1$ to $\A_2$
  is chosen.
  To obtain this partition, we construct a directed graph whose vertices are the SCCs of $\A$. 
  The capacity of each edge is the amount of transitions between the SCCs. 
  The function \texttt{minimum\_cut()} from the \networkx  library \cite{NetworkX} is then used to compute the minimum cut of this graph.
  The source is a fresh vertex with edges to all SCCs with no predecessors, the sink is the last SCC in a topological ordering.
\end{enumerate}

%*******************************************************************************
\vspace{-0.0mm}
\subsection{Implementation of Gate Complementation} \label{sec:impl-gate}
\vspace{-0.0mm}
%*******************************************************************************

\aligater implements the generalized gate complementation methods.
The implementation tries to find all possible partitions of the input NFA $\A$
satisfying the input conditions of any of the methods presented in
\ifTR\cref{sec:gate-general}\else \cref{sec:gate-general-outline}
(cf.~\cite{techrep} for more details)\fi.
The input conditions are formulated as the language equivalence or disjointness of
certain states and evaluated by
\mata; language equivalence is tested using the \emph{antichains algorithm}
\cite{WulfDHR06}, language disjointness is checked via an emptiness check on the product automaton.
The concrete details can be found in
\ifTR{}\cref{sec:impl-gate-details}\else{}\cite{techrep}\fi, but the main idea
is to prefer partitions tagged with \methodEqual to those tagged with
\methodDisjoint, since they usually deliver smaller results, and that we try to
pick a partition where the two components have a~similar size.
%
% From all found partitions (each of them tagged with the method whose
% conditions it satisfies), a~single one is chosen based on the
% following rules:
% %
% \begin{enumerate}
%   \item  First, if there are partitions that do not require intersecting
%     automata during the complementation operation, we discard the rest. \js{nechapu, co se mysli tim intersecting automata...}
%   \item  Second, if there are partitions tagged with \methodEqual, we discard
%     the rest (the construction of \methodEqual has a~higher potential of
%     delivering smaller results).
%   \item  Finally, we pick a~partition with the smallest difference of sizes
%     of~$\A_1$ and~$\A_2$.
% \end{enumerate}
% %
% It is possible that
% no suitable partition is found. In that case, the algorithm aborts.
% \js{pokud toto plati i pro sequential (verim, ze ano), tak by to melo byt formulovano jinak/jinde.}

%%%%%%%%%%%%%%%%%%%%%%%%%%%%%%%%%%%%%%%%%%%%%%%%%%%%%%%%%%%%%%%%%%%
\newcommand{\figScplotFwdRev}[0]{
\begin{figure}[t]
\begin{minipage}{.5\textwidth}
\includegraphics[width=\textwidth]{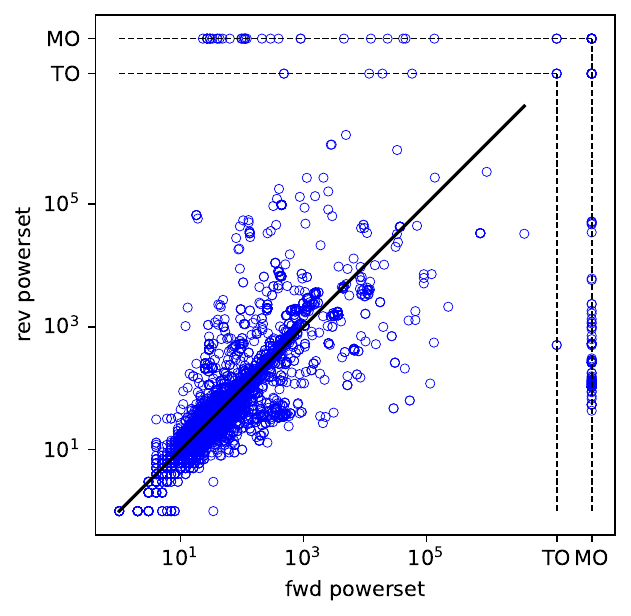}
\end{minipage}
\begin{minipage}{.5\textwidth}
\includegraphics[width=\textwidth]{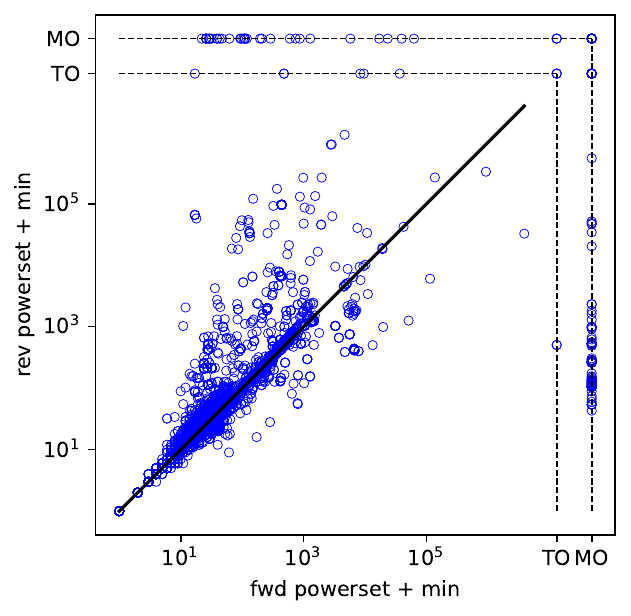}
\end{minipage}
\caption{Comparison of sizes of complements (number of states) generated by 
\fwdpws and \revpws
without (left) or with (right) minimization;
axes are logarithmic.}\label{fig:scplot-fwd-x-rev}
\end{figure}
}

\newcommand{\figScplotSeqGateAtp}[0]{
\begin{figure}[t]
\centering
\begin{minipage}{.4\textwidth}
\includegraphics[width=\textwidth]{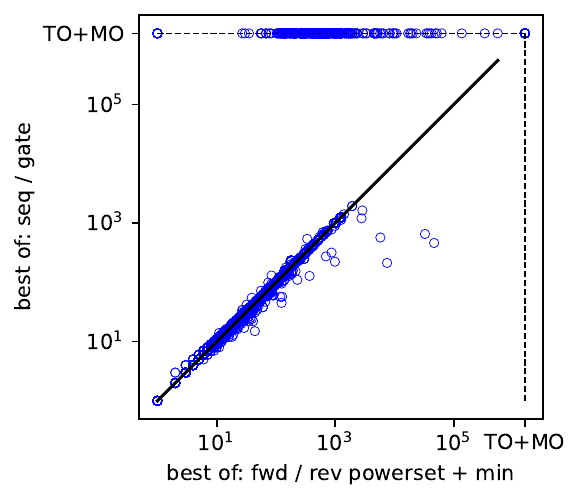}
\end{minipage}
\begin{minipage}{.4\textwidth}
\includegraphics[width=\textwidth]{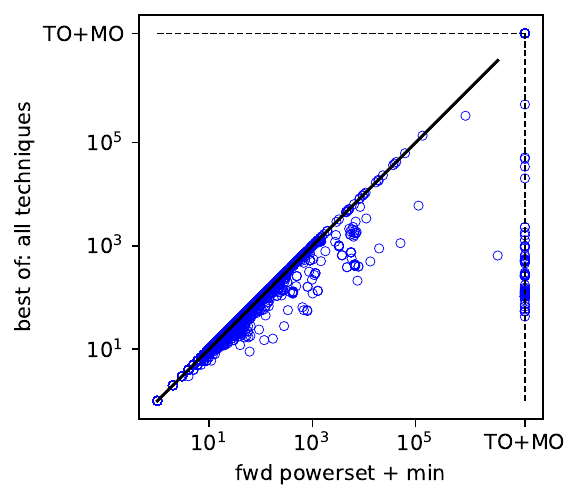}
\end{minipage}
\caption{Comparison of sizes of complements (number of states) given by the
  best of \fwdpws\plusmin and \revpws\plusmin against the best of \sequential and
  \gate (left), and \fwdpws\plusmin against the best result of all techniques (right).
  }\label{fig:scplot-hopcroft-fwdrev-x-seqgate-all}
\vspace{-5mm}
\end{figure}
}

\newcommand{\tabTimeouts}[0]{
\begin{table}[tb]
    \centering
    \caption{The numbers of automata in our benchmark set where the computation
    ran out of resources in relation to the values of $\detsuccsname$.}
    \label{table:heu-blowup}
    \begin{tabular}{crrrr}
    \toprule
     & \multicolumn{2}{c}{\fwdpws} & \multicolumn{2}{c}{\revpws} \\ \cmidrule(lr){2-3}\cmidrule(lr){4-5}
     & ~\quad TO & MO\quad~ & ~\quad TO & MO\quad~ \\
	 \midrule
    $\detsuccs{\A} < \detsuccs{\rev{\A}}$ & 5 & 30\quad~ & 20 & 43\quad~ \\
    $\detsuccs{\A} = \detsuccs{\rev{\A}}$ & 3 & 25\quad~ & 1 & 5\quad~ \\
    $\detsuccs{\A} > \detsuccs{\rev{\A}}$ & 1 & 67\quad~ & 5 & 3\quad~ \\
	\bottomrule
    \end{tabular}
\end{table}
}

\newcommand{\figCactus}[0]{
\begin{figure}[tb]
\includegraphics[width=\textwidth]{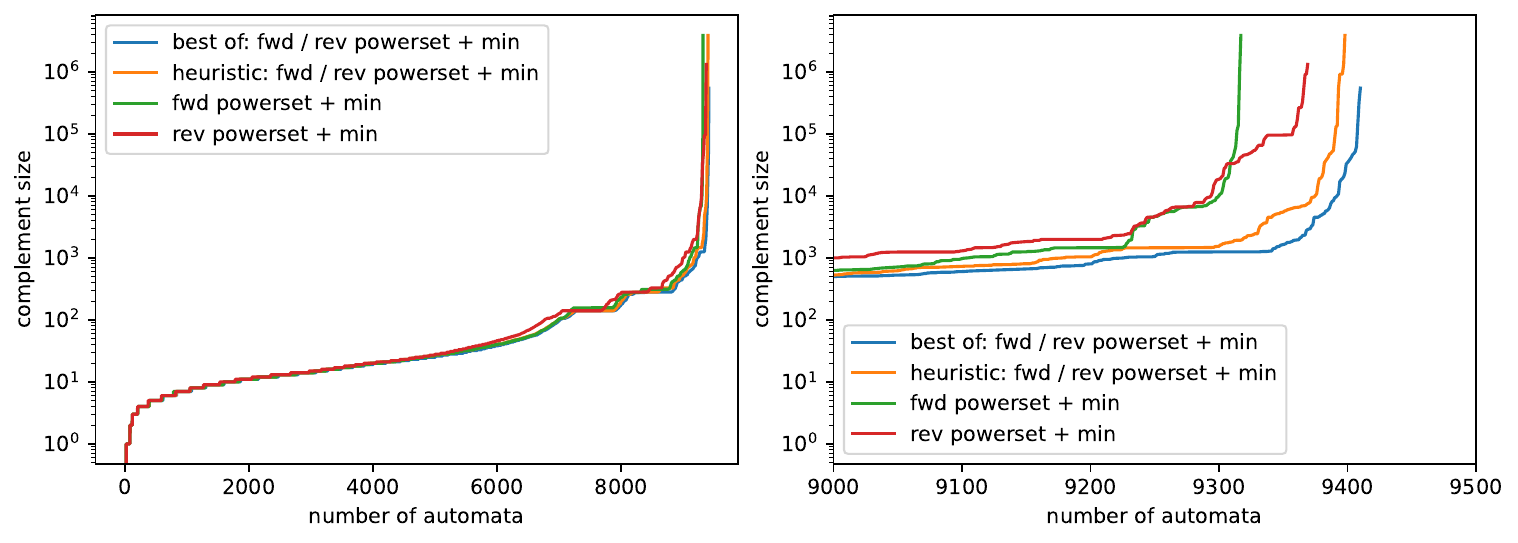}
\caption{A cactus plot comparing the performance of the heuristic with the
individual powerset methods (left), and a zoom into the tail of the plot (right).}
\label{fig:heuristic}
\end{figure}
}

%*******************************************************************************
\vspace{-0.0mm}
\subsection{Evaluation}\label{sec:label}
\vspace{-0.0mm}
%*******************************************************************************

All experiments were run on computers with the Intel\textregistered{}
Core\texttrademark{} i7-8700 CPU. Each complementation algorithm was
executed on each input NFA with the timeout of 5\,min and the memory
limit of 8\,GiB. We focus on the size (number of states) of the
results and also how often the individual methods were successful
(finished within the time and memory limits). TO and MO means that the
time or memory limit was reached, respectively.

For the evaluation, %To evaluate the presented algorithms,
we used a~total of 9,450
benchmarks from \NfaBench~\cite{NfaBench}, which gathers automata benchmarks from diverse applications.
Additional details regarding the families of benchmarks are
in \ifTR\cref{sec:impl-benchmarks}\else\cite{techrep}\fi.

%------------------------------------------------------------------------------
\textbf{Results.} 
% \js{vysvetlit, co znamena \fwdpws\plusmin,\revpws\plusmin, \sequential a \gate}
We evaluated
the performance of the proposed algorithms \sequential and \gate
compared to the best result of \fwdpws\plusmin and \revpws\plusmin
(\cref{fig:scplot-hopcroft-fwdrev-x-seqgate-all}, left).  The methods
\sequential and \gate are, in general, computationally more intensive
than the powerset constructions.  Running \sequential often produces
large automata unless reduced during the process.  The automata
structure, however, allows \rabit~to effectively reduce them, unlike
automata generated by powerset constructions.  Applicability of \gate
is limited by input conditions (it produced results for 2,577
benchmarks) and evaluating these conditions can be costly.  Despite
frequent timeouts, these methods sometimes achieve significantly
better results than powerset-based methods, suggesting further room
for improvement in NFA complementation beyond powerset-based
techniques.

We have also evaluated the benefit of using all available techniques (e.g., in
a~portfolio) against \fwdpws\plusmin as the baseline method
(\cref{fig:scplot-hopcroft-fwdrev-x-seqgate-all}, right).
This use case is targeted at applications where it pays off to obtain as small
automaton for the complement as possible, such as when translating an extended
regex into an NFA that will be used millions of times during matching.
The results show that the proposed 
techniques were in many cases able to bring significant benefits, in particular
solving a~number of cases when \fwdpws\plusmin ran out of resources.

\figScplotSeqGateAtp   %%%%%%%%%%%%%%%%%%%%%

%}}}
%{{{ Conclusion

\vspace{-3mm}
\section{Conclusion}
\vspace{-2mm}

We have presented, to the best of our knowledge, the first systematic
empirical study of NFA complementation approaches. We suggested
several novel algorithms for complementation of (subclasses of) NFAs.
We carried out an extensive experimental evaluation of the approaches
and showed that alternative complementation algorithms can often give
a significantly better result than the classic approach (sometimes
even in orders of magnitude). We have also suggested %and evaluated
a~heuristic that helps to select between the classic approach and
reverse powerset complementation.

% We have presented, implemented, and evaluated a number of algorithms
% for complementation of nondeterministic finite automata.
% % , including two novel algorithms for complementation of subclasses of NFAs.
% Though these methods may require a specific structure of the input
% automaton and are sometimes costly in terms of resource usage, they
% often produce a significantly smaller (even in the orders of
% magnitude) complement than the classic approach based on automata
% determinization by (forward) powerset construction.
% %
% We have also suggested and evaluated a~heuristic that helps to select
% between the classic approach and reverse powerset complementation.

% combination with conservative complementation approach supported by an appropriate heuristic brings significant
% improvement -- in terms of the complement size -- to the classical powerset-based complementation method, even in an order
% of magnitude.

There are still multiple opportunities for improvement,
% in several parts of the described approaches, 
e.g., in the partitioning process for the
sequential complementation. Moreover, other structural classes of
automata amenable for efficient complementation may exist.
%allowing for more efficient complementation.

%}}}

%%%%%%%%%%%%%%%%%%%%%%%%%%%%%%%%%%%%%%%%%%%%%%%%%%%%%%%%%%%%%%%%%%%%%%%%%%%%%%%%
\vspace{-0.0mm}
\section*{Acknowledgements}\label{sec:acknowledge}
\vspace{-0.0mm}
%%%%%%%%%%%%%%%%%%%%%%%%%%%%%%%%%%%%%%%%%%%%%%%%%%%%%%%%%%%%%%%%%%%%%%%%%%%%%%%%

We thank the anonymous reviewers for careful reading of the paper and
their suggestions that improved its quality.  L.~Holík and O.~Lengál
were supported by the Czech Ministry of Education, Youth and Sports
ERC.CZ project LL1908, the Czech Science Foundation projects 23-07565S
and 25-17934S, and the FIT BUT internal project
FIT-S-23-8151. A.~Štěpková and J.~Strejček were supported by the
European Union’s Horizon Europe program under the grant agreement
No.~101087529 (CHESS). 

\vspace{-3mm}
%%%%%%%%%%%%%%%%%%%%%%%%%%%%%%%%%%%% BIBLIO %%%%%%%%%%%%%%%%%%%%%%%%%%%%%%%%%%
\bibliography{literature}
\bibliographystyle{splncs04}
%%%%%%%%%%%%%%%%%%%%%%%%%%%%%%%%%%%% BIBLIO %%%%%%%%%%%%%%%%%%%%%%%%%%%%%%%%%%

\ifTR
\newpage
\appendix
% some bug in cleveref
\crefalias{section}{appendix}
\crefalias{subsection}{appendix}

\input{appendix.tex}

\fi

\end{document}

%% file: appendix.tex
\section{Additional Definitions} \label{sec:ext-nfa-additional}

\newcommand{\induced}[2]{#1|_{#2}}

This section adds a few definitions that were either only outlined or completely omitted in \cref{sec:generalized-compl} due to the page limit.

Let us fix a port NFA $\A = (Q, \Sigma, \delta, \inportsets, \outportsets)$. For a set of states $R\subseteq Q$, the states in $R$ with an incoming transition from $Q\smallsetminus R$ are called \emph{entry ports} and states of $R$ with an outgoing transition to $Q\smallsetminus R$ are called \emph{exit ports}. 
We denote the set of all entry ports of $R$ by $\inportsof{R}$ and the set of exit ports by $\outportsof{R}$.
For a partition $\{R, S\}$ of $Q$, we call the transitions between $R$ and $S$ \emph{transfer transitions} denoted by $\gates$.
If the transfer transitions lead only from $R$ to $S$, then
the ordered pair $(R,S)$ is a \emph{sequential partition}.

A \emph{subautomaton} of the port NFA $\A$ \emph{induced} by a set of states $R\subseteq Q$ isthe port NFA $\induced \A R= (R, \Sigma_R, \delta_R, \inportsets_R, \outportsets_R)$ where
% $\delta_R = \{p \xrightarrow{a} q \in \delta \mid p, q \in R\}$,
$\delta_R=\delta\cap R\times\Sigma\times R$, 
$\Sigma_R = \{a \in \Sigma \mid \exists p \xrightarrow{a} q \in \delta_R\}$ and its port sets are of two types:
(1) \emph{outer port sets} are intersections of the original port sets with $R$ and (2) \emph{inner port sets} contain states of $\inportsof{R}$ and $\outportsof R$, partitioned arbitrarily. Different complementation constructions will use different partitionings. 
Formally, 
$\inportsets_R = (I_1\cap R,\ldots,I_k\cap R)\cdot \innerinportsets$ and 
$\outportsets_R = (F_1\cap R,\ldots,I_{\ell}\cap R)\cdot \inneroutportsets$ where 
$\innerinportsets$ is a sequence of port sets, which are all subsets of $\inportsof R$ and 
$\inneroutportsets$ is a sequence of port sets, subsets of $\outportsof R$.

If $(Q_1, Q_2)$ is a sequential partition of $Q$, we say that $(\A_1, \A_2)$ is a sequential partition of $\A$, where $\A_1 = (Q_1, \Sigma_1, \delta_1, \inportsets_1, \outportsets_1)$ and $\A_2 = (Q_2, \Sigma_2, \delta_2, \inportsets_2, \outportsets_2)$ are the subautomata of $\A$ induced by $Q_1$ and $Q_2$, respectively. We call $\A_1$ and $\A_2$ the \emph{front} and \emph{rear component}. The subautomaton $\A_1$ has a sequence $\inneroutportsets_1$ of inner exit ports, which contain states with outgoing transitions to $\A_2$, while $\A_2$ has a sequence of inner entry ports denoted by $\innerinportsets_2$, gathering states with incoming transitions from $\A_1$.

When addressing the entry and exit port sets of a port NFA, we often abuse notation and write $I_i \in \inportsets$ and $F_j \in \outportsets$. A state is considered \emph{reachable} in a port NFA, if it is reachable in at least one of its slices.

The union of two port NFAs is done element-wise, which enforces that both input automata must have the same amount of entry and exit port sets. Given two port NFAs $\A_1 = (Q_1, \Sigma, \delta_1, (I_{1i})_{0 \leq i \leq k}, (F_{1j})_{0 \leq j \leq \ell})$ and $\A_2 = (Q_2, \Sigma, \delta_2, (I_{2i})_{0 \leq i \leq k}, (F_{2j})_{0 \leq j \leq \ell})$, \\such that $Q_1 \cap Q_2 = \emptyset$, their union is the port NFA $\A_1 \Cup \A_2 = (Q_1 \cup Q_2, \Sigma, \delta_1 \cup \delta_2, (I_{1i} \cup I_{2i})_{0 \leq i \leq k}, (F_{1j} \cup F_{2j})_{0 \leq j \leq \ell})$.

For cases when it is not convenient to refer to a port set by its index, we introduce the following notation. For a port NFA $\A = (Q, \Sigma, \delta, \inportsets, \outportsets)$ and two
sets $P, R \subseteq Q$, the NFA $\autfromto{\A}{P}{R} = (Q, \Sigma, \delta, P,
R)$ is given by assigning $P$, resp.\ $R$ as the set of initial, resp.\ final states. The notation may be extended to states by setting $\autfromto{\A}{p}{q} = \autfromto{\A}{\{p\}}{\{q\}}$ for $p, q \in Q$.

\section{Generalized Sequential Complementation} \label{sec:port-general}

This construction is a generalization of the sequential complementation presented in \cref{sec:seq-simple} and aims to build possibly nondeterministic complements of arbitrary port NFAs. It uses a port NFA $\A$ partitioned into $\A_1$ and $\A_2$, where $\A_1$ is deterministic and complete. If $\A_1$ is not deterministic, it may be determinized as defined in \cref{sec:extended-powerset}.

The basic idea of sequential complementation assumes that the front and rear component are connected by a single edge. Here we aim to extend the technique to the setting where the components are connected by arbitrary transitions in the direction from $\A_1$ to $\A_2$. Also, $\A_1$ may contain outer exit ports (or ``final states''), and $\A_2$ may contain outer entry ports (``initial states''). Other than that, the principle of the basic construction is kept; the complement $\C$ of $\A$ is composed from the deterministic $\A_1$ and the (possibly nondeterministic) complement $\C_2$ of $\A_2$ in a very similar manner.

For the intuition, let us fix $I_i \in \inportsets$ and $F_j \in \outportsets$ and think about $\A$ as if it were a standard NFA $\autij{\A}{i}{j}$. Let $\A$ read a word $w \in \Sigma^*$ starting in the only initial state of $\A_1$ (because $\A_1$ is deterministic). Firstly, there is a run of $\A$ that stays in $\A_1$ its whole life because $\A_1$ is complete. Hence, $\C$ needs to check that $w \notin \langof{\A_1}$. Secondly, consider a different run of $\A$ and the moment where it has just moved from $\A_1$ to an entry port $p$ in $\A_2$. If the suffix $v$ of $w$ is currently remaining on input, we will call $v$ a $p$-\emph{remainder}. In this case, $\C$ will instantiate a fresh copy of $\C_2$, which will check that $v$ is not accepted in $\A_2$ starting from $p$. In fact, $\C$ will do so for every possible pair $(v, p)$, where $v$ is a $p$-remainder. 

\medskip

To formally define the construction, let $\A = (Q, \Sigma, \delta, \inportsets, \outportsets)$, where $\inportsets = (I_i)_{0 \leq i \leq k}$ and $\outportsets = (F_j)_{0 \leq j \leq \ell}$, be a port NFA partitioned into a sequential partition of two subautomata $\A_1 = (Q_1, \Sigma, \delta_1, \inportsets_1, \outportsets_1)$ and $\A_2 = (Q_2, \Sigma, \delta_2, \inportsets_2, \outportsets_2)$, where $\A_1$ is deterministic and complete.

Because we do not need to construct a complement of $\A_1$, it does not need any inner exit port sets, and, therefore, we set $\inneroutportsets_1$ to be an empty sequence. On the contrary, the complement $\C_2$ of $\A_2$ must be able to recognize which words are accepted by $\A_2$ starting from any of its inner entry ports. Therefore, if all inner entry ports of $\A_2$ are indexed from $k$ to $k'$ as $p_{k+1}, \ldots, p_{k'}$, we set $\innerinportsets = (\{p_i\})_{k < i \leq k'}$. 
Given the newly specified entry port sets, let $\C_2 = (\elc{Q}_2, \Sigma, \elc{\delta}_2, \elc{\inportsets}_2, \elc{\outportsets}_2)$ be a complement of $\A_2$.

The complement of $\A$ is the port NFA $\C = (\elc{Q}, \Sigma, \elc{\delta}, \elc{\inportsets}, \elc{\outportsets})$ defined as follows.
\begin{itemize}
    \item Identically to the basic construction, we set $\elc{Q} = Q_1 \times 2^{\elc{Q}_2}$. The meaning also stays the same; in the pair $(q, R) \in \elc{Q}$, $q$ is the current state of $\A_1$, and $R$ contains the current states of the active instances of $\C_2$.
    \item The transition relation $\elc{\delta}$ is again formed by three building blocks: the corresponding transition in $\A_1$, the transitions of each running instance of $\C_2$, and the possible (multiple) new instances of $\C_2$.
    
    % The successors of $R$ in $\C_2$ are given by the \emph{multi-product transition function} $\tprod \colon 2^{\elc{Q}} \times \Sigma \rightarrow 2^{2^{\elc{Q}}}$.
    % \begin{gather*}
    %     \tprod(\{r_1, \ldots, r_n\}, a) = \{ \{s_1, \dots, s_n \} \mid r_i \xrightarrow{a} s_i \in \elc{\delta}_2, 1 \leq i \leq n\}
    % \end{gather*}
    % Note that if at least one $r_i$ has no successors under $a$, the whole set $R$ also has no successors under $a$. In such a case, intuitively, all runs of $\C_2$ over some remainder $v$ represented by $r_i$ are definitely not accepting, so there is no point in continuing with this path.
    
    Consider a state $(q, R) \in \elc{Q}$ and $a \in \Sigma$. The states representing fresh runs of $\C_2$ are determined by the \emph{intergalactic multi-product transition function} $\tprodinter$. It takes every entry port $p$ of $\A_2$ such that $q \xrightarrow{a} p \in \gates$, and lets $\C_2$ run from some state in the port complement of $\{p\}$.
        
    If $\gates(q, a) = \{p_{i_1}, \ldots, p_{i_m}\}$, then
    $$
        \tprodinter(q, a) = \{ \{t_1, \dots, t_m \} \mid t_j \in \elc{I}_{2{i_j}} , 1 \leq j \leq m \}.
    $$

    % Note that if $\gates(q, a)$ is empty, then $\tprodinter(q, a) = \{\emptyset\}$. This corresponds to $q$ not being an exit port and no new runs of $\C_2$ being launched. However, if $\elc{I}_{2{i_j}}$ is empty for some $I_{2i_j} = \{p_{i_j}\}$, then $\tprodinter(q, a) = \emptyset$. The empty $\elc{I}_{2{i_j}}$ means everything is accepted from $p_{i_j}$ in $\A_2$, so all runs of $\C_2$ starting from the port complement of $p_{i_j}$ can be terminated immediately because they will not accept anything anyway.

    The transition function $\elc{\delta}$ is defined as follows. For each $(q,\{r_1,\ldots,r_n\})\in\elc{Q}$, $a\in\Sigma$, transition relations $q\xrightarrow{a}q'\in\delta_1$,$r_i\xrightarrow{a}s_i\in\elc{\delta}_2$ for all $1\le i\le n$, and $T \in \tprodinter(q, a)$, the transition relation $\elc{\delta}$ contains a transition
    $$
    (q,\{r_1,\ldots,r_n\}) \xrightarrow{a} (q',\{s_1,\ldots,s_n\} \cup T).
    $$

    % The successors of $(q, R)$ in $\C$ are then set to be 
    % \begin{gather*}
    %     \elc{\delta}((q,R),a) = \{ (r, S \cup T ) \mid r = \delta_{1}(q,a), S \in \tprod(R,a), T \in \tprodinter(q,a) \}.
    % \end{gather*}

    \item In addition to the basic setting, $\A_2$ may have outer entry ports (or ``initial states''). If that is the case, $\C$ initiates an instance of $\C_2$ running from the corresponding outer entry ports right at the start. Therefore, the port complement of $I_i \in \inportsets$ is defined as follows. Because $\A_1$ is deterministic, we suppose $I_{1i} = \{q_0\}$ (follows from the definition of a deterministic port NFA). 
    $$
        \elc{I}_i = \begin{cases}
            \{ (q_0, \emptyset)\} &\text{if } I_{2i} = \emptyset, \\
            \{ (q_0, \{r_0\}) \mid r_0 \in \elc{I}_{2i} \} &\text{otherwise.}
        \end{cases}
    $$

    \item Contrary to the basic construction, $\A_1$ can have outer exit ports (or ``final states''). Therefore, $\C$ must not accept the words that would be accepted by $\A_1$. Hence, the port complement of $F_j \in \outportsets$ is $\elc{F}_j = (Q_1 \smallsetminus F_{1j}) \times 2^{\elc{F}_{2j}}$.

\end{itemize}

\begin{theorem}
    For each $I_i \in \inportsets$ and $F_j \in \outportsets$, it holds that $\langof{\autij{\C}{i}{j}} = \compl{\langof{\autij{\A}{i}{j}}}.$
\end{theorem}
\begin{proof}
  First, consider $w \in \langof{\autij{\A}{i}{j}}$. 

  If $w \in \langof{\autij{\A_1}{i}{j}}$, every run of $\autij{\C}{i}{j}$ over $w$ ends in some state $(q, R)$, where $q \in F_{1j}$, and so $(q, R) \notin \elc{F}_{j}$.

  If $w \in \langof{\autij{\A_2}{i}{j}}$, then $I_{2i} \neq \emptyset$, and $w \notin \langof{\autij{\C_2}{i}{j}}$. If $\elc{I}_{2i}$ is not empty, there is a state $(q_0, \{r_0\}) \in \elc{I}_i$, where $I_{1i} = \{q_0\}$ and $r_0 \in \elc{I}_{2i}$. Hence, each state $(q, R)$ reached by $\autij{\C}{i}{j}$ reading $w$ is not accepting, as $R$ must contain a state $r \notin \elc{F}_{2j}$, in which a run of $\autij{\C_2}{i}{j}$ over $w$ ends. If, on the other hand, $\elc{I}_{2i}$ is empty, then also $\elc{I}_i = \emptyset$ and $w$ is certainly not accepted by $\autij{\C}{i}{j}$. 

  If $w \notin \langof{\autij{\A_1}{i}{j}}$ and $w \notin \langof{\autij{\A_2}{i}{j}}$, it must be the case that $w = ucv$, such that $\autij{\A}{i}{j}$ reaches an exit port $p$ of $\A_1$ over $u$, passes a transition $p \xrightarrow{c} p'$ into an entry port $p'$ of $\A_2$, and a run of $\A_2$ starting in $p'$ accepts $v$ in $F_j$. In other words, $v$ is a $p'$-remainder, $\{p'\} = I_{2i'}$ for some $k < i' \leq k'$, and $v \in \langof{\autij{\A_2}{i'}{j}}$. Hence,
  $\autij{\C}{i}{j}$ can reach only states of the form $(p,R)$ after reading
  $u$. When $\autij{\C}{i}{j}$ reads
  $c$ from this state, there are two options. Either $\elc{I}_{2i'} = \emptyset$, then $\tprodinter(p, c) = \emptyset$ and, therefore, $\C$ has no transitions from $(p,R)$. 
  If $\elc{I}_{2i'} \neq \emptyset$, the complement $\autij{\C}{i}{j}$ reaches a state $(p',R')$ from $(p,R)$ under $c$, where
  $R'$ has to contain some entry port $t \in \elc{I}_{2i'}$ of
  $\C_2$. However, $v\in \langof{\autij{\A_2}{i'}{j}}$ implies that $\autij{\C_2}{i'}{j}$ does not accept
  $v$. Hence, each state $(q'',R'')$ of $\autij{\C}{i}{j}$ reached from $(p',R')$ by reading
  $v$ is not accepting as it cannot satisfy
  $R''\subseteq\elc{F}_{2j}$. To sum up, $\autij{\C}{i}{j}$ has no accepting run over $ucv=w$.

  Now assume that $w \notin \langof{\autij{\A}{i}{j}}$. As
  $\autij{\A_1}{i}{j}$ is deterministic and complete, it has a single run over $w$. Then $w \notin \langof{\autij{\A_1}{i}{j}}$, and, therefore, the run of $\autij{\A_1}{i}{j}$ over $w$ ends in $q \in Q_1 \smallsetminus F_{1i}$. It also holds that $w \notin \langof{\autij{\A_2}{i}{j}}$, which implies that $\autij{\C_2}{i}{j}$ has an accepting run over $w$. 
  
  Whenever the run of $\autij{\A_1}{i}{j}$ reaches an exit port $p$ of $\A_1$ over some prefix $u$ of $w$ and $w$ is of the form $ucv$, we know that for every transition $p \xrightarrow{c} p'$ into an entry port $p'$ of $\C_2$, $v$ is not accepted in $\autij{\A_2}{i}{j}$ starting from $p'$. We may assume that $\{p'\} = I_{2i'}$ for some $k < i' \leq k'$, as $p'$ is an entry port if $\A_2$. In other words, if the prefix
  $u$ is followed by $c$, then
  $\autij{\C_2}{i'}{j}$ has an accepting run over the corresponding suffix
  $v$. We can construct an accepting run of
  $\autij{\C}{i}{j}$ over $w$ as follows. If $I_{2i}$ is not empty, we set the initial state of this run as $(q_0, \{r_0\})$, where $r_0 \in \elc{I}_{2i}$ is the initial state of the accepting run of $\autij{\C_2}{i}{j}$. If $I_{2i} = \emptyset$, the initial state of our accepting run is the only initial state of $\autij{\C}{i}{j}$, $(q_0, \emptyset)$.
  Whenever the automaton
  $\autij{\A_1}{i}{j}$ tracked in the first element of the states of
  $\C$ reaches an exit port $p$ and $\autij{\C}{i}{j}$ reads
  $c$, then for every transition $p \xrightarrow{c} p' \in \gates$, we add to the second element of the state of
  $\C$ the initial state of the corresponding accepting run of
  $\autij{\C_2}{i'}{j}$. We follow each chosen accepting run of $\C_2$ in the future transitions of
  $\autij{\C}{i}{j}$. After reading the whole
  $w$, the second element of the reached state of
  $\C$ will contain only accepting states of
  $\elc{F}_{2j}$. Thus, the constructed run of $\autij{\C}{i}{j}$ over
  $w$ is accepting and $w\in\langof{\autij{\C}{i}{j}}$.
\end{proof}

\section{Generalized Gate Complementation} \label{sec:gate-general}

This complementation method is an extension of the idea presented in \cref{sec:gate-simple}. It considers a subset of NFAs called \emph{gate} automata and utilizes specific properties of these automata to create nondeterministic complements. We give two distinct generalizations of the basic technique from \cref{sec:gate-simple}, called \methodEqual and \methodDisjoint. Both methods are applicable to a restricted class of gate NFAs. The respective classes for both methods are incomparable. At the end of this section, we also present several modifications of the two main methods, which further extend the classes of automata we are able to complement.

Let $\A = (Q, \Sigma, \delta, \inportsets, \outportsets)$ be a port NFA partitioned into $\A_1 = (Q_1, \Sigma_1, \delta_1, \inportsets_1, \outportsets_1)$ and $\A_2 = (Q_2, \Sigma_2, \delta_2, \inportsets_2, \outportsets_2)$. Let $\gsym = \{c \mid \exists p_1 \xrightarrow{c} p_2 \in \gates\}$ denote the set of symbols appearing on the transitions between $\A_1$ and $\A_2$. 

The partition $(\A_1, \A_2)$ is a \emph{gate partition} of $\A$, if $\gsym \cap \Sigma_i = \emptyset$ for at least one $i \in \{1, 2\}$. If $\A$ has a gate partition, it is a \emph{gate automaton}. The transitions in $\gates$ are then called \emph{gates}, and $\gsym$ is a set of \emph{gate symbols}.

To distinguish entry and exit ports by the symbols on their gates, we define the following notions. For $c \in \gsym$, we set $\outsymgrp{1}{c} = \{p_1 \in \outportsof{Q_1} \mid \exists p_2 \in Q_2 ~.~ p_1 \xrightarrow{c} p_2 \in \gates\}$, and similarly, $\insymgrp{2}{c} = \{p_2 \in \inportsof{Q_2} \mid \exists p_1 \in Q_1 ~.~ p_1 \xrightarrow{c} p_2 \in \gates\}$. Note that a port $p$ can belong to multiple such sets if it has transitions under multiple different symbols.

\begin{figure}[t]
    \centering
    \begin{tikzpicture}[smallautomaton,node distance=8mm,every state/.style={minimum size=5mm}]

    \tikzset{component/.style={rounded corners=10pt, fill=blue!20, draw=blue!20, minimum width=2cm, minimum height=2cm, semithick}}
    
    \node[initial, component, label={[blue, label distance =-0.6cm]90:\large $\A_1$}] (A1) {};
    \node[state, accepting, double=blue!15, double distance=1pt, anchor=east, xshift = -3mm] (fin1) at (A1.east) {};
    
    \node[component, label={[blue, label distance =-0.6cm]90:\large $\A_2$}, right = of A1] (A2) {};
    \node[state, accepting, double=blue!15, double distance=1pt, anchor=east, xshift = -3mm] (fin2) at (A2.east) {};
    
    \node[coordinate, yshift=4mm] (1) at (A1.east) {};
    \node[coordinate] (2) at (A1.east) {};
    \node[coordinate, yshift=-4mm] (3) at (A1.east) {};
    \node[coordinate, yshift=4mm] (4) at (A2.west) {};
    \node[coordinate] (5) at (A2.west) {};
    \node[coordinate, yshift=-4mm] (6) at (A2.west) {};
    
    \node[anchor=south] (A1sym) at (A1.south) {$\Sigma \smallsetminus \gsym$};
    \node[anchor=south] (A2sym) at (A2.south) {$\Sigma$};
     
    \path[->]
    (1) edge node {} (4)
    (2) edge node {} (5)
    (3) edge node [below] {$\gsym$} (6)
    ;
    \end{tikzpicture}
    \caption{A scheme of a gate NFA with no gate symbols in $\A_1$, accepting states symbolize possible occurrences of exit ports.}
    \label{fig:gate-first}
\end{figure}

From now on, let us consider only gate automata with no gate symbols in $\A_1$. The presented ideas can be easily translated to gate NFAs with no gate symbols in $\A_2$, as will be discussed later. We also suppose there are no outer entry ports in $\A_2$, that is $I_{2i} = \emptyset$ for each $0 \leq i \leq k$, and show later how this constraint can be removed. \cref{fig:gate-first} shows a scheme of a gate NFA satisfying these constraints.

In such a gate NFA, the first gate symbol appearing in a word $w \in \Sigma^*$ must be read when transitioning through a gate. For now, let us fix $I_i \in \inportsets$, $F_j \in \outportsets$ and refer to $\A$ as if it was a standard NFA $\autij{\A}{i}{j}$. Suppose $w = ucv$, where $u \in (\Sigma \smallsetminus \gsym)^*$, $c\in \gsym$, and $v \in \Sigma^*$. If $w \notin \langof{\A}$, then one of the following happens when $\A$ reads $w$: 
\begin{enumerate}
    \item $\A_2$ is not reached by any run of $\A$. This means $u$ is not accepted by any inner exit port of $\A_1$ with a gate under $c$.
    \item Some runs of $\A$ may reach an inner entry port of $\A_2$ after reading $uc$, but $v$ is not accepted in any run of $\A_2$ starting from the given entry ports.
\end{enumerate}

When building a complement of a port gate NFA $\A$, we follow these two cases. We construct two automata $\pre{\C}$ and $\suf{\C}$, where $\pre{\C}$ accepts all words falling under point (1), as well as all $w \in \compl{\langof{\A}} \cap (\Sigma \smallsetminus \Gamma)^*$. The automaton $\suf{\C}$ accepts all words from point (2). The complement of $\A$ is then defined as $\C = \pre{\C} \Cup \suf{\C}$. Below we show the construction of $\pre{\C}$, which is the same for both \methodEqual and \methodDisjoint method. The difference between these methods lies in the construction of $\suf{\C}$, which depends on specific properties of $\A$. We leave the exact definitions of the construction of $\suf{\C}$ for the following sections and only show a scheme of the complement $\C$ for both methods in \cref{fig:gate-comp-general}.

The automaton $\pre{\C}$ is built as follows. If we index the gate symbols from $\ell$ to $\ell'$ as $\gsym = \{c_{\ell + 1}, \ldots, c_{\ell'}\}$, the inner exit port sets of $\A_1$ are set as $\inneroutportsets_1 = (\outsymgrp{1}{c_{j'}})_{{\ell} < j' \leq \ell'}$. In other words, each inner port set corresponds to one gate symbol and consists of ports with an outgoing gate under said symbol. We then use the complement $\C_1 = (\elc{Q}_1, \Sigma \smallsetminus \gsym, \elc{\delta}_1, \elc{\inportsets}_1, \elc{\outportsets}_1)$ of $\A_1$ with respect to the alphabet $\Sigma \smallsetminus \gsym$, and set $\pre{\C} = (\pre{Q}, \Sigma, \pre{\delta}, \pre{\inportsets}, \pre{\outportsets})$, where
\begin{itemize}
    \item $\pre{Q} = \elc{Q}_1 \cup \{ s \}$,
    \item $\pre{\delta} = \elc{\delta}_1 \cup \{p \xrightarrow{c_{j'}} s \mid p \in \elc{F}_{1j'}, \ell < j' \leq \ell' \} \cup \{s \xrightarrow{a} s \mid a \in \Sigma\}$.
    \item $\pre{\inportsets} = (I_{pre(i)})_{0 \leq i \leq k}$, and for $I_i \in \inportsets$, we set 
    $I_{pre(i)} = I_{1i}$. Note that we suppose $I_{2i} = \emptyset$.
    \item $\pre{\outportsets} = (F_{pre(j)})_{0 \leq j \leq \ell}$, and for $F_j \in \outportsets$, we define $F_{pre(j)} = \elc{F}_{1j} \cup \{s\}$.
\end{itemize}

\begin{lemma}
    \label{lemma:gate-pre}
    For each $I_i \in \inportsets$, $F_j \in \outportsets$, and $w \in \Sigma^*$, the following points hold.
    \begin{enumerate}
        \item $\compl{\langof{\autij{\A}{i}{j}}} \cap (\Sigma \smallsetminus \gsym)^* \subseteq \langof{\autij{\pre{\C}}{i}{j}}$
        
        % If $w \in (\Sigma \smallsetminus \gsym)^*$ and $w \notin \langof{\autij{\A}{i}{j}}$, then $w \in \langof{\autij{\pre{\C}}{i}{j}}$.
        \item If $w = ucv$ such that $u \in (\Sigma \smallsetminus \gsym)^*$, $c = c_{j'} \in \gsym$, $v \in \Sigma^*$, and $u \notin \langof{\autij{\A_1}{i}{j'}}$, then $w \in \langof{\autij{\pre{\C}}{i}{j}}$.
        
        \item $\langof{\autij{\A}{i}{j}} \cap \langof{\autij{\pre{\C}}{i}{j}} = \emptyset$
        
        % If $w \in \langof{\autij{\A}{i}{j}}$, then $w \notin \langof{\autij{\pre{\C}}{i}{j}}$.
    \end{enumerate}
\end{lemma}
\begin{proof}
    \begin{enumerate}
        \item If $w \notin \langof{\autij{\A}{i}{j}}$ and $w \in (\Sigma \smallsetminus \gsym)^*$, then $w \notin \langof{\autij{\A_1}{i}{j}}$ and $w \in \langof{\autij{\C_1}{i}{j}}$. Because $w$ contains no gate symbols, $s$ is never reached when reading $w$, and thus no run of $\pre{\autij{\C}{i}{j}}$ ends in $\elc{F}_{1j} \cup \{s\}$.
        
        \item Because $u \in \langof{\autij{\C_1}{i}{j'}}$, then after reading $uc$, $\autij{\pre{\C}}{i}{j}$ transitions from $\elc{F}_{1j'}$ to $s$ and $w$ is eventually accepted.
        
        \item  Let us have $w \in \langof{\autij{\A}{i}{j}}$. If $w$ contains no gate symbols, it is surely not accepted in $s$. But because it contains no gate symbols, then $w \in \langof{\autij{\A_1}{i}{j}}$, and it is neither accepted by $\autij{\C_1}{i}{j}$. If $w = ucv$, where $u \in (\Sigma \smallsetminus \gsym)^*$ and $c = c_{j'} \in \gsym$, then $w$ is certainly not accepted in $\autij{\C_1}{i}{j}$, because it contains a gate symbol. But because $u \in \langof{\autij{\A_1}{i}{j'}}$, then $u \notin \langof{\autij{\C_1}{i}{j'}}$ and state $s$ is never reached when $\autij{\pre{\C}}{i}{j}$ reads $w$.
          \qedhere
    \end{enumerate}
\end{proof}

\subsection{Equal Gate NFAs}

The simplest and potentially smallest construction of $\suf{\C}$ tries to reuse the construction from \cref{sec:gate-simple} with only minor changes, but extend it to a larger class of automata. More precisely, consider $w = ucv$, where $u \in (\Sigma \smallsetminus \gsym)^*$, $c\in \gsym$, and $v \in \Sigma^*$. This version of $\suf{\C}$ just checks that $v$ is not accepted in $\A_2$ starting in any of its entry ports with a gate under $c$ and does not consider the prefix $uc$. However, for $\pre{\C} \Cup \suf{\C}$ to be a complement of $\A$, the following additional constraints must be put on $\A$.

Intuitively, it must not matter through which port was $\A_2$ entered for any prefix $uc$. In other words, all inner entry ports of $\A_2$ must be reached through gates under $c$ by exactly the same prefixes.

This constraint is formally defined as follows. For each gate symbol $c \in \gsym$, we set the subautomaton $\A_{1c}$ of $\A$ as the automaton $\A_1$ together with all gates under $c$, defined as $\A_{1c} = (Q_1 \cup \insymgrp{2}{c}, \Sigma_1 \cup \{c\}, \delta_{1c}, \inportsets_1, \outportsets_1)$, where $\delta_{1c} = \delta_1 \cup \{ p_1 \xrightarrow{c} p_2 \mid p_1 \xrightarrow{c} p_2 \in \gates \}$. Then, if for every gate symbol $c \in \gsym$ and for every entry port set $I_i \in \inportsets_{1}$, we have $\insymgrp{2}{c} = \{ p_1, \ldots , p_n \}$ and 
$\langof{\autfromto{\A_{1c}}{I_i}{\insymgrp{2}{c}}} = \langof{\autfromto{\A_{1c}}{I_i}{p_1}} = \ldots = \langof{\autfromto{\A_{1c}}{I_i}{p_n}}$, 
we say $(\A_1, \A_2)$ is an \emph{equal gate partition}. If $\A$ has an equal gate partition, it is called an \emph{equal gate NFA}.

If $\A$ is an equal gate NFA, $\suf{\C}$ is constructed as follows. Similarly as for $\A_1$ in the construction of $\pre{\C}$, we index the gate symbols as $\gsym = \{c_{k+1}, \ldots, c_{k'}\}$, and set the inner entry port sets of $\A_2$ to be $\innerinportsets_2 = (\insymgrp{2}{c_{i'}})_{k < i' \leq k'}$. Each inner port set again corresponds to one gate symbol and consists of ports into which there is a gate under said symbol. Refer to \cref{fig:gate-comp-general} for a scheme of $\suf{\C}$.

If $\C_2 = (\elc{Q}_2, \Sigma, \elc{\delta}_2, \elc{\inportsets}_2, \elc{\outportsets}_2)$ is a complement of $\A_2$ with respect to the alphabet $\Sigma$, we set $\suf{\C} = (\suf{Q}, \Sigma, \suf{\delta}, \suf{\inportsets}, \suf{\outportsets})$, where:
\begin{itemize}
    \item $\suf{Q} = \elc{Q}_2 \cup \{ t \}$,
    \item $\suf{\delta} = \elc{\delta}_2 \cup \{ t \xrightarrow{c_{i'}} p \mid c_{i'} \in \gsym, p \in \elc{I}_{2i'} \} \cup \{ t \xrightarrow{a} t \mid a \in \Sigma \smallsetminus \gsym \}$.
    \item $\suf{\inportsets} = (I_{\mathit{suf}(i)})_{0 \leq i \leq k}$, and $I_{\mathit{suf}(i)} = \{t\}$ for $I_i \in \inportsets$, 
    \item $\suf{\outportsets} = (F_{\mathit{suf}(j)})_{0 \leq i \leq \ell}$, and $F_{\mathit{suf}(j)} = \elc{F}_{2j}$ for $F_j \in \outportsets$.
\end{itemize}

\begin{lemma}
    \label{lemma:gate-equal-suf}
    If $\A$ is an equal gate NFA, for each $I_i \in \inportsets$, $F_j \in \outportsets$, and $w \in \Sigma^*$, the following holds.
    \begin{enumerate}
        \item If $w = ucv$, such that $u \in (\Sigma \smallsetminus \gsym)^*$, $c = c_{i'} \in \gsym$, and $v \in \Sigma^* \smallsetminus \langof{\autij{\A_2}{i'}{j}}$, then $w \in \langof{\autij{\suf{\C}}{i}{j}}$.
        \item $\langof{\autij{\A}{i}{j}} \cap \langof{\autij{\suf{\C}}{i}{j}} = \emptyset$
    \end{enumerate}
\end{lemma}
\begin{proof}
    \begin{enumerate}
        \item When $\autij{\suf{\C}}{i}{j}$ reads $w = ucv$, it transitions from $t$ to $I_{2i'}$ only exactly after reading $uc$. Because $v \notin \langof{\autij{\A_2}{i'}{j}}$, then $v \in \langof{\autij{\C_2}{i'}{j}}$, and is, therefore, accepted in $\elc{F}_{2j} = F_{\mathit{suf}(j)}$ by $\autij{\suf{\C}}{i}{j}$.
        
        \item Suppose $w \in \langof{\autij{\A}{i}{j}}$. If $w$ contains no gate symbols, the only run of $\suf{\C}$ over $w$ stays in $t$ and is nonaccepting. If $w = ucv$, where $u \in (\Sigma \smallsetminus \gsym)^*$ and $c = c_{i'} \in \gsym$, then certainly $v \in \langof{\autij{\A_2}{i'}{j}}$ and $v \notin \langof{\autij{\C_2}{i'}{j}}$. Because $\autij{\suf{\C}}{i}{j}$ transitions from $t$ to $I_{2i'}$ only after reading $uc$ when reading $w$, it does not accept $w$.
      \qedhere
    \end{enumerate}
\end{proof}

Now, recall that the complement of $\A$ is $\C = \pre{\C} \Cup \suf{\C}$.

\begin{theorem}
    \label{lemma:gate-equal}
    If $\A$ is an equal gate NFA, then for all $I_i \in \inportsets$ and $F_j \in \outportsets$, we have $\langof{\autij{\C}{i}{j}} = \compl{\langof{\autij{\A}{i}{j}}}$.
\end{theorem}
\begin{proof}
    From \cref{lemma:gate-pre,lemma:gate-equal-suf} it follows that $\langof{\autij{\A}{i}{j}} \cap \langof{\autij{\C}{i}{j}} = \emptyset$. \cref{lemma:gate-pre} also ensures that words in $\compl{\langof{\autij{\A}{i}{j}}} \cap (\Sigma \smallsetminus \gsym)^*$ are accepted by $\autij{\C}{i}{j}$.
    
    Let $w = ucv \notin \langof{\autij{\A}{i}{j}}$, such that $u \in (\Sigma \smallsetminus \gsym)^*$, $c \in \gsym$, and $v \in \Sigma^*$. If $c = c_{j'}$, such that $u \notin \langof{\autij{\A_1}{i}{j'}}$, \cref{lemma:gate-pre} ensures that $w \in \langof{\autij{\C}{i}{j}}$. Otherwise, because $\A$ is an equal gate NFA, it must hold that $c = c_{i'}$ such that $v \notin \langof{\autij{\A_2}{i'}{j}}$, and we can use \cref{lemma:gate-equal-suf} to say that $w \in \langof{\autij{\C}{i}{j}}$.
\end{proof}

\subsection{Disjoint Gate NFAs}

The constraints on the input NFA enforced by the equal gate condition are quite strict. Suppose we do not require the equality of $\A_2$'s entry-port languages and give $\suf{\C}$ information about which gates $\A$ uses when reading a word $w \in \Sigma^*$ (see \cref{fig:gate-comp-general} for a scheme). Each inner entry port of $\A_2$ is treated individually instead of grouping the ports by the symbols on their gates as before.

Suppose $w = ucv$, where $u \in (\Sigma \smallsetminus \gsym)^*$, $c\in \gsym$, and $v \in \Sigma^*$, and let us again view $\A$ as a standard NFA. If an entry port $p$ of $\A_2$ is reached after reading $uc$, then $\suf{\C}$ checks that $v$ cannot be accepted from $p$ in $\A_2$. If that is the case for any $p$ that is reached by reading $uc$, $\suf{\C}$ accepts $w$. However, consider $w$ such that when reading it, $\A$ can pass through two gates $p_1 \xrightarrow{c} p_2$ and $p'_1 \xrightarrow{c} p'_2$, where $p_2 \neq p'_2$. If $v$ is accepted by $\A_2$ from $p_2$, then $w \in \langof{\A}$. However, if $v$ is not accepted by $\A_2$ starting in $p'_2$, then $w$ is accepted by $\suf{\C}$, even though it belongs to $\langof{\A}$. This motivates an additional input condition for $\A$.

We say $(\A_1, \A_2)$ is a \emph{disjoint gate partition} of $\A$, if for all $c \in \gsym$, $0 \leq i \leq k$, $0 \leq j \leq \ell$, and every two gates $p_1 \xrightarrow{c} p_2, p'_1 \xrightarrow{c} p'_2 \in \gates$, we have 
$$\langof{\autfromto{\A_1}{I_{1i}}{p_1}} \cap \langof{\autfromto{\A_1}{I_{1i}}{p'_1}} \neq \emptyset \implies \langof{\autfromto{\A_2}{p_2}{F_{2j}}} = \langof{\autfromto{\A_2}{p'_2}{F_{2j}}}.$$ 
This means that if there is a possibility to reach $p_1$ and $p'_1$ under the same prefix $u$, both $p_2$ and $p'_2$ represent the same languages in $\A_2$. If $\A$ has a disjoint gate partition, it is a \emph{disjoint gate NFA}.

Suppose $\A$ is a disjoint gate NFA. The inner entry port sets of $\A_2$ are one-element sets containing individual inner entry ports of $\A_2$. If we index the inner entry ports of $\A_2$ as $\{p_{k+1}, \ldots, p_{k'}\}$, then the inner entry port sets of $\A_2$ are $\innerinportsets_2 = (\{p_{i'}\})_{k < i' \leq k'}$.

The automaton $\suf{\C} = (\suf{Q}, \Sigma, \suf{\delta}, \suf{\inportsets}, \suf{\outportsets})$ is composed from $\A_1$
and the complement of $\A_2$ with respect to $\Sigma$, which is an automaton
$\C_2 = (\elc{Q}_2, \Sigma, \elc{\delta}_2, \elc{\inportsets}_2, \elc{\outportsets}_2)$, as follows.
\begin{itemize}
    \item $\suf{Q} = Q_1 \cup \elc{Q}_2$
    \item $\suf{\delta} =  \delta_1 \cup \elc{\delta}_2 \cup 
    \{ p \xrightarrow{c} p' \mid p \xrightarrow{c} p_{i'} \in \gates \text{ and } p' \in \elc{I}_{2i'} \}$
    \item $I_{\mathit{suf}(i)} = I_{1i}$ for $I_i \in \inportsets$ (note that $I_{2i} = \emptyset$).
    \item $F_{\mathit{suf}(j)} = \elc{F}_{2j}$ for $F_j \in \outportsets$.
\end{itemize}

\begin{lemma}
    \label{lemma:gate-disjoint-suf}
    If $\A$ is a disjoint gate NFA, for each $I_i \in \inportsets$ and $F_j \in \outportsets$, and $w \in \Sigma^*$, the following holds.
    \begin{enumerate}
        \item If $w = ucv$, such that $w \notin \langof{\autij{\A}{i}{j}}$, $u \in (\Sigma \smallsetminus \gsym)^*$, $c = c_{j'} \in \gsym$, $v \in \Sigma^*$, and $u \in \langof{\autij{\A_1}{i}{j'}}$ then $w \in \langof{\autij{\suf{\C}}{i}{j}}$.
        \item $\langof{\autij{\A}{i}{j}} \cap \langof{\autij{\suf{\C}}{i}{j}} = \emptyset$
    \end{enumerate}
\end{lemma}
\begin{proof}
    \begin{enumerate}
        \item Consider any entry port $p_{i'} \in Q_2$ that $\A$ reaches after reading $uc$. Because $w \notin \langof{\autij{\A}{i}{j}}$, it must hold that $v \notin \langof{\autij{\A_2}{i'}{j}}$, from which follows $v \in \langof{\autij{\C_2}{i'}{j}}$. Because $\elc{I}_{2i'}$ can be reached in $\suf{\C}$ under $uc$ in the same way $p_{i'}$ was reached, we have $w \in \langof{\autij{\C}{i}{j}}$.
         
        \item Suppose $w \in \langof{\autij{\A}{i}{j}}$, and note that $\suf{\C}$ may transfer from $\A_1$ to $\C_2$ only when reading the first gate symbol appearing in $w$. If $w$ contains no gate symbols, then all runs of $\suf{\C}$ over $w$ stay in $\A_1$ and do not accept. 
        
        Now suppose that $w = ucv$, where $u \in (\Sigma \smallsetminus \gsym)^*$, $c \in \gsym$, and $v \in \Sigma^*$. Then $\A$ must be able to pass through a gate $q \xrightarrow{c} p_{i'} \in \gates$ such that $v \in \langof{\autij{\A_2}{i'}{j}}$. Let $r \xrightarrow{c} p_{i''} \in \gates$ be a gate through which $\A$ passes in a different run over $w$. Because both $q$ and $r$ were reached when reading $u$ from $I_{1i}$, by the definition of disjoint gate NFAs we get $\langof{\autij{\A_2}{i'}{j}} = \langof{\autij{\A_2}{i''}{j}}$. That yields $v \notin \langof{\autij{\C_2}{i'}{j}} = \langof{\autij{\C_2}{i''}{j}}$, which means $w \notin \langof{\autij{\suf{\C}}{i}{j}}$.
      \qedhere
    \end{enumerate}
\end{proof}

\begin{theorem}
    \label{lemma:gate-disjoint}
    If $\A$ is a disjoint port gate NFA, then for all $I_i \in \inportsets$ and $F_j \in \outportsets$, we have $\langof{\autij{\C}{i}{j}} = \compl{\langof{\autij{\A}{i}{j}}}$.
\end{theorem}
\begin{proof}
    Implied by \cref{lemma:gate-pre,lemma:gate-disjoint-suf}.
\end{proof}

\subsection{Modifications of Presented Methods}\label{sec:modifications}
Returning to the note in the beginning of \cref{sec:gate-general}, now we show how some of the restrictions placed on the input automata may be changed.

Let us consider a port gate NFA $\A$ with no gate symbols in the second component. If we also require that $\A_1$ has no outer exit ports, we can apply the presented methods in a ``reversed'' manner. Note that it may not be sufficient to simply run gate complementation on $\rev{\A}$ because while $\A$ may have an equal or disjoint gate partition, $\rev{\A}$ may not, although it is still a gate NFA. So, by introducing the following changes, we are able to complement more automata with gate complementation than if only the reverse operation was used.

In the construction of $\pre{\C}$, $\A_1$ is complemented with respect to the alphabet $\Sigma$ and $s$ receives loops under $\Sigma \smallsetminus \gsym$. The exit ports of $\C_1$ are not added as exit ports of $\pre{\C}$, which means we assign $F_{pre(j)} = \{s\}$ for each $F_j \in \outportsets$.

For $\suf{\C}$, $\A_2$ is complemented with respect to $\Sigma \smallsetminus \gsym$. For equal gate NFAs, $t$ gets loops under $\Sigma$. And lastly, the entry ports of $\C_2$ are added to the entry ports of $\suf{\C}$. Formally, for each $I_i \in \inportsets$, we set $I_{\mathit{suf}(i)} = \{t\} \cup \elc{I}_{2i}$ for equal gate NFAs, and $I_{\mathit{suf}(i)} = I_{1i} \cup \elc{I}_{2i}$ for disjoint gate NFAs. 

\medskip

If $\A$ is a gate NFA with no gate symbols in $\A_1$, but has outer entry ports in $\A_2$, a word $w$ containing gate symbols can be accepted either in a run starting in $A_1$ and passing through a gate into $A_2$, or in a run starting in an initial state of $A_2$. In this case, the complement of $\A$ is the intersection of the following two automata:
\begin{enumerate}
    \item gate complement of $\A$, where the outer entry ports of $\A_2$ are removed before the complementation
    \item complement of $\A_2$ with its entry ports (constructed by any method)
\end{enumerate}

If, instead, there are no gate symbols in $\A_2$ and $\A_1$ has outer exit ports at the same time, the process is very similar; the gate complement of $\A$ is intersected with the complement of $\A_1$.

\section{Subexponential Upper Bound for a Restricted NFA Class in Sequential Complementation}\label{sec:seq-upper-bound}

This section aims to give a non-exponential upper bound for sequential complementation applied to a specific class of automata. This upper bound targets the general version of the algorithm (\cref{sec:port-general}), but is also applicable to the simplified construction (\cref{sec:seq-simple}).

The sequential construction is prone to explosion when tracking too many simultaneous instances of the complement $\C_2$ of the rear component. Although the resulting complements' size can be in practice substantially reduced with minimization algorithms, there is a class of automata for which the sequential composition produces small results as-is. In particular, we target the automata where the sequential complement tracks at most one active instance of $\C_2$ at all times. We do not claim that this condition encapsulates all automata for which sequential complementation can give small results; we only identify a class of automata for which it is simple enough to see the validity of the given upper bound.

In the following, we assume a complement $\C = (\elc{Q}, \Sigma, \elc{\delta}, \elc{\inportsets}, \elc{\outportsets})$ defined in \cref{sec:port-general}, constructed from a deterministic and complete front component $\A_1$ and a complement $\C_2$ of the rear component $\A_2$.

\begin{theorem}
  If every state $(q, R) \in \elc{Q}$ satisfies $|R| \leq 1$, then $|\C| \leq |\A_1| + n|\C_2|$, where $n$ is the number of states in $\A_1$ to which there is a transition from $\A_1$'s inner exit port. Moreover, if $\A_2$ is deterministic or reverse-deterministic, then $|\C| \in \mathcal{O}(|\A|^2)$.
\end{theorem}
\begin{proof}
  There are at most $|\A_1|$ states of the form $(q, \emptyset)$, where $q \in Q_1$.

  For every transition $(q, \{r\}) \xrightarrow{a} (q', \{r'\}) \in \elc{\delta}$, the transition $r \xrightarrow{a} r'$ exists in $\C_2$. Therefore (and because $\A_1$ is deterministic), the states $(q, R)$ with $|R| = 1$ will form several (possibly interconnected) copies of $\C_2$, each copy differing from the other only in the first element $q$ of each state. The different first elements may be caused by several instances of $\C_2$ being activated when $\A_1$ is in different states. As the amount of these states of $\A_1$ is $n$ as stated in the theorem and because $\A_1$ is deterministic, there are at most $n|\C_2|$ states of the form $(q, \{r\})$, with $q \in Q_1$ and $r \in \elc{Q_2}$.

  If $\A_2$ is deterministic or reverse-deterministic, then $|\C_2| \in \mathcal{O}(|\A_2|)$. Therefore, $|\C| \in \mathcal{O}(|\A_1| + n|\A_2|) \subseteq \mathcal{O}(|\A_1| + |\A_1|\cdot|\A_2|) \subseteq \mathcal{O}(|\A|^2)$.
\end{proof}

\begin{theorem}
  Every state $(q, R) \in \elc{Q}$ satisfies $|R| \leq 1$ if all of the following holds:
  \begin{enumerate}
    \item If $p$ is an inner exit port of $\A_1$ with a transfer $a$-transition to $\A_2$, no inner exit port in $\A_1$ is reachable from $p$ by a path starting with an $a$-transition.
    \item For every inner exit port $p$ of $\A_1$ and $a \in \Sigma$, there is at most one $a$-transition from $p$ to $\A_2$.
    \item $\A_2$ has no outer exit ports.
  \end{enumerate}
\end{theorem}
\begin{proof}
  Condition 3 ensures that $R$ is empty for any entry port $(q, R)$ of $\C$. Condition 2 ensures that $\tprodinter$ adds at most one new state to $R$ at a time, for any $(q, R) \in \elc{Q}$. Due to Condition 1, once a new instance of $\C_2$ is activated by $\tprodinter$, no other instance of $\C_2$ will be activated in the same run of $\C$, as the underlying run of $\A_1$ will not pass through another inner exit port. As $\tprod$ does not increase the number of states in $R$, the size of $R$ in any state $(q, R)\in \C$ will be either zero or one.
\end{proof}

Note that in the simple setting in \cref{sec:seq-simple}, Conditions 2 and 3 are trivially satisfied.

Moreover, the given upper bound is strictly better than that of forward powerset for the same class of automata. The example in \cref{fig:port-simple-ex} and its generalization described in \cref{sec:seq-simple} satisfy the given conditions, but have exponential deterministic complements.

\section{Implementation Details}

\vspace{-0.0mm}
\subsection{Gate Complementation} \label{sec:impl-gate-details}
\vspace{-0.0mm}
%*******************************************************************************

\aligater implements the generalized gate complementation methods.
The implementation tries to find all possible partitions of the input NFA $\A$
satisfying the input conditions of any of the methods presented in \cref{sec:gate-general}. The input conditions are formulated as the language equivalence or disjointness of
certain states and evaluated by
\mata; language equivalence is tested using the \emph{antichains algorithm}
\cite{WulfDHR06}, language disjointness is checked via an emptiness check on the product automaton.

From all found partitions (each of them tagged with the method whose
conditions it satisfies), a~single one is chosen based on the
following rules:
\begin{enumerate}
  \item  First, if there are partitions that do not require intersecting
		automata during the complementation operation (cf.\
		\cref{sec:modifications}), we discard the rest.
  \item  Second, if there are partitions tagged with \methodEqual, we discard
    the rest (the construction of \methodEqual has a~higher potential of
    delivering smaller results).
  \item  Finally, we pick a~partition with the smallest difference of sizes
    of~$\A_1$ and~$\A_2$.
\end{enumerate}
It is possible that
no suitable partition is found. In that case the algorithm aborts.
(Note that for sequential complementation, the algorithm always finds
a~suitable partition, which might be trivial.)
% % From all found partitions (each of them tagged with the method whose
% % conditions it satisfies), a~single one is chosen based on the
% % following rules.  Preferred are the partitions that do not require
% % intersecting automata in the complementation, and \methodEqual is
% % chosen over \methodDisjoint, as the construction has a higher
% % potential of delivering smaller results.  Lastly, the partitions where
% % both components are similar in size are favored.  It is possible that
% % no suitable partition is found. In that case the algorithm aborts.
%
% % Gate complementation requires complements of both $\A_1$ and $\A_2$. The implementation tries both forward and reverse powerset for each component, and takes the smaller result.

\subsection{Evaluation Settings} \label{sec:impl-eval-settings}

The following settings of \aligater were used in the experiments:
\begin{itemize}
    \item \fwdpws and \revpws implement the respective powerset comple\-men\-ta\-tions
      without any minimization within the process.

    \item
      $x$\plusmin denotes that we ran method~$x\in\{\fwdpws,\revpws\}$ while applying Hopcroft's
      minimization~\cite{Hopcroft71,ValmariL08} (implemented within
      \mata) immediately after the powerset construction.

    \item \sequential runs sequential
      complementation %(\cref{sec:seq-simple} \ol{ref appendix?})
      with the following settings:
    \begin{itemize}
        \item all three partitioning approaches described in
          \cref{sec:impl-seq} were used and the smallest result was
          returned,
        \item the last component was complemented using reverse powerset complementation, and
        \item every intermediate result was reduced by the function
          \texttt{reduce()} from \mata, which uses simulations to reduce
          automata~\cite{BustanG03}.
    \end{itemize}

  \item \gate runs gate complementation, where both front and rear
    components are complemented by both forward and reverse powerset
    complementation and the smaller complement is selected for
    further processing. % (\cref{sec:gate-simple}\ol{ref appendix?}).

    \item The following settings were common for both \sequential and \gate:
    \begin{itemize}
        \item minimization was applied after every powerset construction,
        \item the result was reduced by \rabit~\cite{MayrC13}, with the
          \texttt{lookahead} parameter set to 10, and
        \item we ran the algorithm on both the original input NFA and its
          reverse (which was reversed back after the algorithm) and kept the
          smaller output.
    \end{itemize}

  \item Unreachable and dead states (i.e., states from which no
    accepting state is reachable) were removed at the end of each
    algorithm.
\end{itemize}

\subsection{Benchmark Details} \label{sec:impl-benchmarks}

From \NfaBench, we used NFAs from the directories
\texttt{regexps/regexps\_distinct} and \texttt{regexps/regexps\_union} 
(regular expressions used in network intrusion detection and
protocol identification, 6,624 benchmarks),
\texttt{automata\_inclusion} (abstract regular model checking, 272 benchmarks),
\texttt{email\_filter} (regular expressions, 75 benchmarks),
\texttt{presburger-explicit} (decision procedure for Presburger arithmetic, 174 benchmarks),
\texttt{z3-noodler} (string solving, 2,040 benchmarks; for the
\texttt{z3-noodler/intersection} subdirectory, we used the first 1,000 out of
the 17,000 benchmarks), and
\texttt{ws1s} (decision procedure for WS1S, 265 benchmarks).
The benchmarks originate from various sources
(\cite{BouajjaniHHTV08,HeizmannHP13,Sutcliffe2017,FiedorHJLV17,ChenCHHLS24,HabermehlHHHL24}
and \url{https://regexlib.com/}).

\section{Additional Experiments} \label{sec:additional-exp}

We compared the effect of \fwdpws and \revpws, with and without minimization.
The results are given in \cref{fig:scplot-fwd-x-rev}.
It shows that neither \fwdpws nor \revpws
is clearly superior to the other method
% achieves significantly better results than the other
and that their strengths are complementary.
Although \fwdpws seems
to perform slightly better than \revpws when minimization is employed,
it is important to note that \revpws\plusmin ran out of resources in fewer cases~(79)
than \fwdpws\plusmin~(132).

\figScplotFwdRev  %%%%%%%%%%%%%%%%%%%%%%%%%%%

\cref{fig:scplot-input-x-fwdmin} shows how difficult it is to complement our benchmarks. There are numerous cases where the minimal forward powerset automaton is many orders of magnitude larger than the input automaton, or even where the computation ran out of resources. However, the majority of the collected benchmarks (5,459 in total) had minimal deterministic complements of the same or smaller size than the original automaton.

\begin{figure}
\centering
\begin{minipage}{.5\textwidth}
\includegraphics[width=\textwidth]{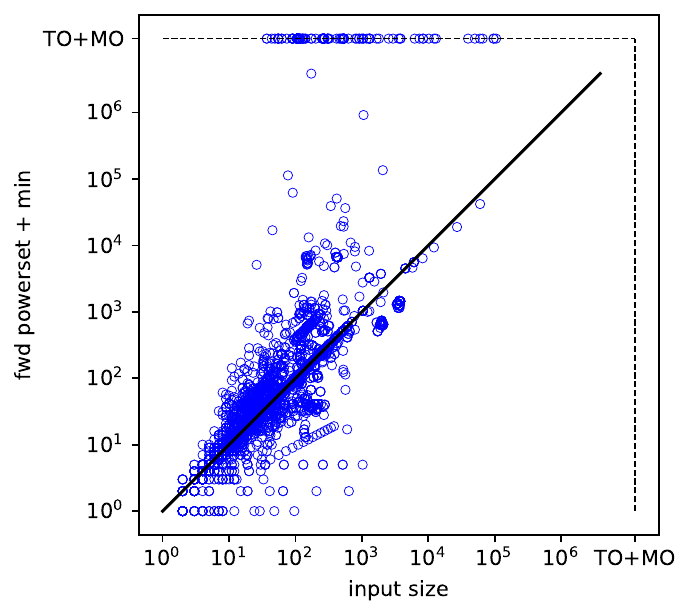}
\end{minipage}
\caption{Comparison of the size of the input NFAs and their minimal deterministic complements.}\label{fig:scplot-input-x-fwdmin}
\end{figure}

In addition to \cref{fig:scplot-fwd-x-rev}, \cref{fig:scplot-fwd-x-fwdrev} shows
how combining forward and reverse powerset complementation
brings significant improvement compared to the classical forward complementation,
again without (left) and with (right) minimization.

\begin{figure}
\begin{minipage}{.5\textwidth}
\includegraphics[width=\textwidth]{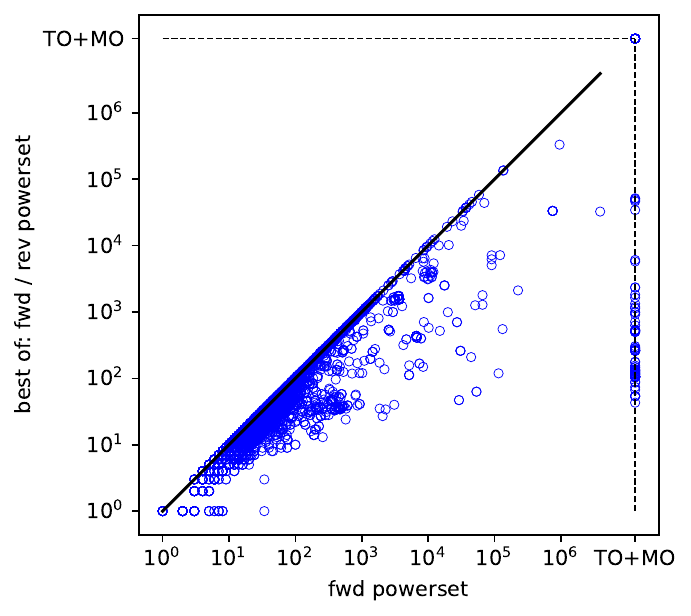}
\end{minipage}
\begin{minipage}{.5\textwidth}
\includegraphics[width=\textwidth]{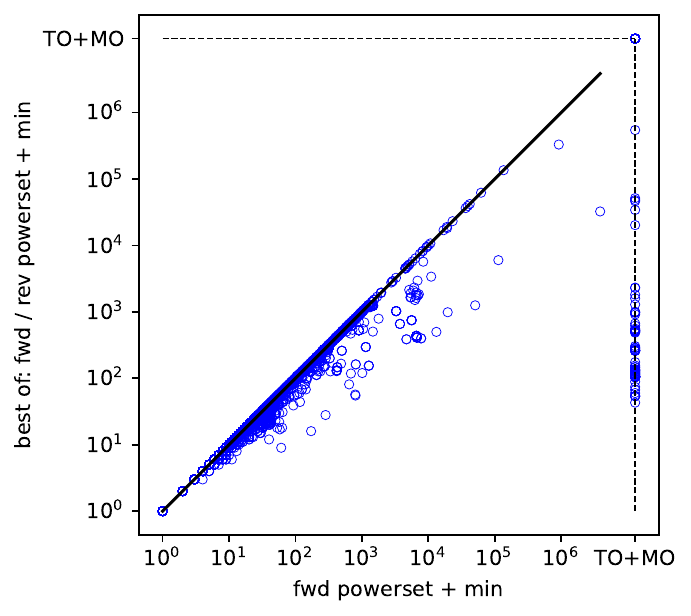}
\end{minipage}
\caption{Comparison of \fwdpws with best of \fwdpws and \revpws
without (left) or with (right) minimization.}\label{fig:scplot-fwd-x-fwdrev}
\end{figure}

Moreover, we have evaluated the effect of the heuristic for deciding between
forward and reverse powerset, as described in \cref{sec:reverse-powerset}.
\cref{table:heu-blowup} shows that the heuristic indeed effectively predicts
resource exhaustion.
% , especially for reverse powerset.

\tabTimeouts  %%%%%%%%%%%%%%%%%%%%%%%%%%%

\cref{fig:heuristic} evaluates the heuristic's performance by
comparing complement sizes of \fwdpws\plusmin, \revpws\plusmin, the
best of both, and the heuristic's choice (if
$\detsuccs{\A} \geq \detsuccs{\rev{\A}}$, choose \revpws\plusmin,
otherwise select \fwdpws\plusmin). The heuristic outperforms both
individual methods and approaches the effectiveness of a~portfolio that runs
both methods and picks the smaller output.

\figCactus  %%%%%%%%%%%%%

%}}}